\documentclass[12pt]{article}

\usepackage[authoryear]{natbib}
\usepackage{amsmath,amssymb,amsthm}
\usepackage{geometry}
\usepackage{times}
\usepackage{setspace}
\usepackage{chngcntr}  
\usepackage{enumitem}
\usepackage{dirtytalk}
\numberwithin{equation}{section}
\numberwithin{table}{section}
\numberwithin{figure}{section}

\newtheorem{assumption}{Assumption}
\newtheorem{lemma}{Lemma}
\newtheorem{example}{Example}
\newtheorem{proposition}{Proposition}
\newtheorem{definition}{Definition}
\newtheorem{remark}{Remark}
\newtheorem{corollary}{Corollary}
\newcommand{\proofinapp}[1]{\noindent\textit{Proof.} See Appendix~\ref{#1}.}
\usepackage[colorlinks=true,
            linkcolor=blue,
            citecolor=blue,
            urlcolor=blue]{hyperref}
\begin{document}

\title{Evaluation and Assignment with Networked Competition and Spillovers\thanks{Cabrales gratefully acknowledges funding from AEI under grants PID2024-156629NB-I00 and CEX2021-001181-M.
 }}
\author{Antonio Cabrales\thanks{%
Department of Economics, Universidad Carlos III de Madrid; e-mail:
antonio.cabrales@uc3m.es}, Wenhao Cheng\thanks{%
Corresponding author. Research Institute of Economics and Management, Southwestern University of Finance and Economics; email:
chengwh@swufe.edu.cn}}
\maketitle

\begin{abstract}
This paper studies how organizations should jointly design evaluation rules and assign workers when performance depends on both effort and non-discretionary advantage. Agents choose effort in positions linked by a competition network, while their effective advantage depends on own type and spillovers through a second network. The planner chooses both the assignment and the effort weight in evaluation. Equilibrium effort rises with a position’s Katz-Bonacich centrality and falls with effective advantage. The optimal evaluation rule generally differs from true output. When effort is more important in production, the planner lowers the effort weight and uses negative assortative assignment to strengthen incentives. When advantage is more important, the planner raises the effort weight and uses positive assortative assignment to exploit spillovers. We also study a constraint requiring assignments to be pairwise stable, which creates an output loss depending on the intensity of competition.

%This paper studies how organizations should jointly design evaluation rules and assign workers when performance depends on both effort and non-discretionary advantages. Agents are assigned to positions, choose effort, and are evaluated relative to benchmarks generated by a competition network. Their effective advantage depends on their own type and on spillovers from other positions, summarized by a second network. The planner chooses both an assignment and the weight placed on effort in evaluation. We show that equilibrium effort is increasing in a position's Katz-Bonacich centrality in the competition network and decreasing in its effective advantage. The optimal policy generally makes evaluation differ from true output. When effort is sufficiently important in production, the planner lowers the weight on effort in evaluation and uses negative assortative assignment to strengthen incentives. When effective advantage is sufficiently important, the planner raises the weight on effort and uses positive assortative assignment to exploit spillovers while mitigating the effort-reducing effect of advantage. We also study a constraint requiring assignments to be pairwise stable. Stability sorts agents by private positional advantage rather than aggregate spillover value, creating an output loss whose comparative statics depend on the intensity of competition.
\end{abstract}



\newpage

\onehalfspacing
\section{Introduction}

Relative performance evaluation is central to how organizations allocate promotions and rewards. Workers are often compared with peers, either explicitly through rankings, predefined distributions (\citealp{CardinaelsFeichter2021}), calibration meetings (\citealp{KhanKornWilliams2024}), or implicitly through managerial assessments of who is outperforming whom.\footnote{A large literature in economic theory has studied the design of relative performance incentives, from rank-order tournaments (\citealp{LazearRosen1981RankOrder}) to correlated outcomes (\citealp{Fleckinger2012}), peer effects (\citealp{Krakel2016}), and contractual externalities (\citealp{OzdenorenYuan2017}).} Yet performance in organizations is rarely generated by effort alone. It also reflects non-discretionary advantages such as ability, fit, reputation, project quality, mentoring, or access to better colleagues. Moreover, these advantages are not fixed independently of organizational design: they may spill over through learning, collaboration, and knowledge transmission, and they depend on where agents are assigned.\footnote{ \citet{JohnsenKuSalvanes2024}, for example, provide evidence from Norway of implicit tournaments around visible behaviors such as not taking leave. This makes a contest model with ex post reward allocation a natural framework for such environments.}

This paper studies how evaluation and assignment should be designed jointly in such environments. We introduce a model with two networks. A competition network determines which positions are compared in evaluation, while a spillover network determines how agents’ fixed types translate into effective advantage across positions. The planner chooses both an evaluation rule, which weights effort relative to effective advantage, and an assignment of agents to positions. Agents then choose effort in a relative-performance contest. The analysis shows that the planner generally does not want evaluation to coincide with true output. Instead, the planner may deliberately overweight effort or overweight effective advantage, depending on the importance of effort in output and on the structure of competition and spillovers.

To be more precise, agents in the model are linked through effort-driven competition and advantage-enhancing spillovers, captured by two overlapping networks, which we refer to as the \emph{competition structure} and the \emph{spillover structure}. For example, two salespeople serving similar customer groups may be benchmarked against each other in evaluation, and their talent may affect each other through observation and learning. We refer to the non-discretionary component as \emph{effective advantage}. It is generated by agents' fixed types and by assignment-dependent spillovers across positions. Thus, effective advantage need not reflect a worker's own innate ability alone, but may capture reputation, mentoring, access to better colleagues, or other advantages created by where the worker is assigned. We think this is a conceptual contribution of the model. The planner in the first stage seeks to maximize output, which depends on both total effort and total effective advantage linearly, and does so by both assigning agents to positions and choosing the relative weight placed on effort and effective advantage in evaluation. In the second stage, after observing their assigned positions, effective advantages, and the evaluation rule, agents choose effort to maximize their probability of winning a prize given a constant marginal cost of effort.

To understand our intended applications, consider a professional-service firm that assigns associates to clients or projects. Associates are evaluated relative to their peers, but output depends not only on effort, such as hours, initiative, or client contact, but also on client importance, partner mentoring, and the knowledge generated by nearby colleagues. The firm can shape incentives in two ways. It can decide how much evaluation should reward observable effort rather than realized outcomes that partly reflect assignment-driven advantages. It can also assign workers to projects that differ in their exposure to competitors and in the spillovers they generate for others. Similar issues arise in sales organizations, hospitals, academic departments, schools, and public agencies, where performance comparisons, assignment, and spillovers interact.

In terms of results, we first show that, under reasonable conditions, agents exert effort that is increasing in their Katz centrality in the competition structure and decreasing in their effective advantage. Building on this result, we then characterize the planner’s optimal policy in the first stage and show that evaluation generally does not coincide with output. Interestingly, when the weight of effort in output is above a cutoff, the planner strategically assigns a lower weight to effort in evaluation, equivalently a higher weight to effective advantage, and chooses an assignment that minimizes total effective advantage. By contrast, when the weight of effort in output is below that cutoff, the planner strategically assigns a higher weight to effort in evaluation, equivalently a lower weight to effective advantage, and chooses an assignment that maximizes total effective advantage. In the former case, effective advantage is relatively unimportant for output, so the planner is willing to sacrifice it and instead create a low-advantage environment in which agents work harder to make up for their disadvantage. In the latter case, effective advantage is more important for output, so the planner seeks to maximize total advantage, while reducing its importance in evaluation in order to prevent high-advantage agents from exerting too little effort.

We also extend the model to allow for endogenous reassignment after the planner’s decision, capturing environments in which agents may privately swap tasks. We study this issue through a notion of pairwise stability under transferable utility, so the planner’s assignment must rule out mutually beneficial swaps. We show that, under a symmetric spillover structure, stability imposes a simple sorting restriction on assignments. We then characterize the resulting stability loss in output and show how it depends on the intensity of competition, which is measured by aggregate Katz centrality. Stronger competition reduces the stability loss when effective advantage is relatively important for output, by allowing evaluation to better compensate for assignment distortions. By contrast, stronger competition increases the stability loss when effort is relatively important for output, by magnifying the incentive role of the blocked assignment channel.

Our framework implies that organizational design should not treat evaluation, assignment, and workplace interactions as separate problems. When performance reflects both effort and non-discretionary factors, and when agents are linked by both competition and spillovers, a planner should jointly choose who is placed where and how much evaluation weights effort relative to advantage. The framework also suggests that organizational incentives are shaped not only by rewards, but also by the structure of rivalry and spillovers across positions.

The model also delivers a set of comparative statics that connect the theory to observable variation across organizations and evaluation systems. Equilibrium effort is higher in positions that are more central in the competition network. This means that  agents who are compared with, or influence the standing of, more other agents have stronger incentives to exert effort. By contrast, effort is lower for agents with greater effective advantage, since advantage partly substitutes for effort in determining success. This disincentive effect is weaker when the evaluation rule places more weight on effort. Thus reforms that increase the effort weight in evaluation should raise effort especially among agents who previously enjoyed large non-discretionary advantages.

The assignment implications depend on the gap between the true production technology and the evaluation rule. When evaluation puts more weight on effort than true output does, the planner assigns high-ability agents to positions with high aggregate spillover value. When evaluation puts less weight on effort than true output does, the planner instead uses negative assortative assignment to strengthen incentives for otherwise disadvantaged agents. The  cutoff between these regimes falls as competition becomes more intense, so organizations with denser benchmarking networks are more likely to rely on incentive-generating negative sorting. In contrast, when spillover opportunities are more unequal, the planner is more likely to use positive sorting to exploit the high social value of assigning able agents to central positions. Finally, when assignments must be pairwise stable, observed sorting may reflect private positional advantage rather than aggregate spillover value, creating a measurable wedge between decentralized assignments and the planner's preferred allocation.

Our paper is related to several strands of the literature. First, it connects to assignment and matching models in organizations. \citet{Sattinger1993AssignmentModels} surveys assignment models of earnings and productive heterogeneity, while \citet{BeckerMurphy1992DivisionLabor} and \citet{Garicano2000Hierarchies} emphasize the role of specialization, knowledge, and coordination in the allocation of workers to tasks. \citet{GabaixLandier2008CEOPay} provides a prominent example of a talent-to-position assignment model. Relative to this literature, our contribution is to study assignment jointly with the design of evaluation rules when performance depends both on effort and on non-discretionary advantage.

Second, the paper relates to the literature on performance measurement and distorted incentives. \citet{Baker1992IncentiveContracts} shows that incentive contracts based on imperfect performance measures need not implement first-best behavior, and \citet{HolmstromMilgrom1991Multitask} show how measured performance can distort effort allocation in multi-task environments. \citet{FelthamXie1994PerformanceMeasure} formalize the idea of performance-measure congruity, while \citet{BakerGibbonsMurphy1994SubjectivePerformance} study the interaction between objective and subjective performance measures. In our model, the evaluation rule can deliberately depart from true output, not because the available measure is imperfect, but because the planner uses evaluation to shape effort incentives under competition and spillovers.

Third, the spillover component of the model relates to work on peer effects and knowledge spillovers. \citet{MasMoretti2009Peers} and \citet{CornelissenDustmannSchoenberg2017PeerEffects} provide evidence on peer effects in the workplace, while \citet{Waldinger2012PeerEffectsScience} and \citet{AzoulayGraffZivinWang2010Superstar} study spillovers among scientists and collaborators. \citet{CalvoArmengolPatacchiniZenou2009PeerEffects} show how network structure matters for peer effects in education. Our model abstracts from the empirical identification of peer effects and instead studies how a planner should assign agents when such spillovers interact with relative evaluation.

The paper also contributes to the literature on contests, tournaments, and network games. \citet{LazearRosen1981RankOrder} provide the canonical tournament model, and \citet{MoldovanuSela2001OptimalPrize} study prize allocation in contests. \citet{Konrad2009Strategy} surveys the broader contest literature. Our network structure is also related to \citet{ballester2006s}, who show how equilibrium actions in network games depend on centrality, and to \citet{matros2024contests}, who analyze contests on networks. The distinctive feature of our approach is that the organization chooses both the assignment of agents to positions and the evaluation rule that determines the relevant contest incentives.

Our work relates to the contract-theory and mechanism-design literature on relative performance evaluation, as already mentioned. Early work studies how peer outcomes can be used to filter common shocks and reduce incentive costs \citep{GreenStokey1983,NalebuffStiglitz1983,LazearRosen1981RankOrder}. Subsequent work extends this analysis to richer settings with correlated performance or risks \citep{Mookherjee1984,Bartling2011,Fleckinger2012,OzdenorenYuan2017}. More recently, \citet{SunZhao2024} study a setting with both effort spillovers in production and correlation in risk with detailed network structure. Other work studies how relative performance evaluation can improve performance in more applied organizational settings \citep{Dye1992,MaskinQianXu2000}. Our work adds to this literature by jointly studying evaluation and assignment in a networked environment.

Finally another related strand is the literature on contests and tournaments on networks. Although not explicitly in a network context, \cite{Krakel2016} studies how competition between peers changes agents' incentives. Other works develop explicit contest or competition models on networks and characterize how alliances, rivalries, local network structure, and agents’ perceptions of their rivals affect equilibrium effort and success probabilities \citep{konig2017networks, bozbay2018contest, matros2024contests, BochetFaureLongZenou2021}. Relative to this literature, our framework adds a spillover layer of non-discretionary factors to competition and studies how a planner can intervene in this setting.

The remainder of the paper is organized as follows. Section~\ref{sec:model} presents the model. Section~\ref{analysis} characterizes the second-stage equilibrium and studies the planner’s optimal choice of evaluation and assignment. Section~\ref{sec:pairwise} extends the analysis to pairwise stable assignments and characterizes the resulting stability loss. Section~\ref{sec:conclusion} concludes. 
\section{Model}
\label{sec:model}

We consider an organization with $n$ positions and $n$ agents. The planner assigns agents to positions and chooses how performance is evaluated. Agents then exert effort, and rewards are determined by relative performance. Rewards in this framework need not be purely monetary, but can also be understood as a higher probability of promotion, recognition by a superior, or simply the utility from outperforming peers. These incentives are provided ex post, after evaluation, rather than specified ex ante in explicit contracts.

\subsection{Timing and performance score}
\label{subsec:timing}

The interaction has two stages. In Stage 1, the planner chooses an assignment of agents to positions together with an evaluation rule. In Stage 2, after observing these choices, agents choose effort.

Formally, let
\[
J:=\{1,\dots,n\}
\]
denote the set of positions, indexed by $j,k\in J$, and let
\[
I:=\{1,\dots,n\}
\]
denote the set of agents, indexed by $i\in I$. An assignment is a bijection
\[
m:I\to J,
\]
where $m(i)$ is the position occupied by agent $i$, and $m^{-1}(j)$ is the agent assigned to position $j$.

Agents occupying positions compete on the basis of an internal performance score. The performance score at position $j$ under assignment $m$ is
\begin{equation}
x_j^{(m)}:=\omega e_j^{(m)}+(1-\omega)\bigl(\theta_j^{(m)}+\alpha\bigr),
\label{eq:evaluation_score}
\end{equation}
where $e_j^{(m)}\ge 0$ is the effort exerted at position $j$. $\theta_j^{(m)}$ denotes the effective advantage at position $j$, and $\alpha$ is a common organization-wide shifter, which could also be made
position-specific or agent-specific, without affecting the main results of the model. The parameter $\omega\in(0,1]$ governs how much weight the planner places on effort relative to the non-discretionary component $\theta_j^{(m)}+\alpha$ of performance.

Effective advantage depends on the assignment as it spills over across positions. We capture this by assuming
\begin{equation}
\theta^{(m)}:=Ma^{(m)},
\label{eq:theta_def}
\end{equation}
where $M$ is an entrywise nonnegative matrix and $a^{(m)}$ is the vector of agent types in position order. $a_i>0$ (or equivalently $a_j^{(m)}$ for the position $j=m(i)$ assigned to agent $i$) is broadly a fixed, non-discretionary type, which may capture not only ability, but also other factors such as fit, reputation, or favoritism. The element $M_{jk}$ measures how much the type assigned to position $k$ contributes to the effective advantage at position $j$. We refer to $M$ as \emph{spillover structure}.

The spillover structure $M$ can be interpreted as capturing how fixed agent characteristics propagate across positions. A useful example is a company that sells internet and communication services, with salespeople assigned to different but related customer segments, such as individual households, business clients, and public-sector clients. In this setting, an agent's type may be interpreted as sales talent, which can be thought of as fixed and predetermined, at least during the relevant competition and evaluation period. A more talented salesperson may be better at reading demand, have better instinct for which customers are promising, or understand more clearly which sales strategies matter. These advantages can spill over to nearby positions because salespeople in related segments can observe and learn from them over time. The extent of this spillover depends on the position. For example, teams serving business clients may interact more with teams serving public-sector clients, while teams focused on a narrow niche segment may generate much less spillover. Hence, placing a highly talented salesperson in a more connected segment allows more colleagues to benefit. 

Such spillover, however, is not necessary for the model. In an extreme case, one can let $M_{ij}=0$ for all $i\neq j$, and $M_{jj}=1$ for all $j$, so that $\theta_j^{(m)}=a_j^{(m)}$ is interpreted simply as one's own talent or some other fixed characteristic.

For later use, define
\begin{equation}
b:=M^\top \mathbf{1},
\qquad
b_j=\sum_{r=1}^n M_{rj},
\label{eq:b_def}
\end{equation}
so that aggregate effective advantage can be written as
\[
\mathbf{1}^\top \theta^{(m)}=b^\top a^{(m)}.
\]
Hence $b_j$ measures the contribution of the agent assigned to position $j$ to total effective advantage.

In Stage 1, the planner chooses both the assignment $m$ and an evaluation weight $\omega\in(0,1]$. In Stage 2, given $(m,\omega)$, each agent $i$ chooses effort for the position she occupies.

\paragraph{Empirical counterparts.} The objects $a_i$, $M$, and $b$ have natural empirical counterparts. The type $a_i$ can be measured by pre-period productivity, credentials, experience, or baseline performance. The matrix $M$ can be constructed from observable channels through which fixed characteristics spill over across positions, such as co-location, mentoring links, client overlap, project interdependence, or communication networks. The induced object \[ b_j=\sum_{r=1}^n M_{rj} \] then measures the aggregate spillover value of assigning a high-type agent to position $j$. Thus, the theory predicts not only which agents should work harder, but also how agents should be assigned across positions with different spillover centrality.

\subsection{Stage 2: agents, evaluation, and payoffs}
\label{subsec:stage2}

Throughout this subsection, we take the planner's Stage-1 choices $(m,\omega)$ as given.

Success at position $j$ occurs when evaluated performance exceeds a benchmark formed from the scores of connected positions:
\begin{equation}
x_j^{(m)}+u_j
\geq
\beta\sum_{k=1}^n w_{jk}x_k^{(m)},
\qquad
u_j\sim N(0,\sigma^2)\ \text{i.i.d.}
\label{eq:success_condition}
\end{equation}
where $u_j$ captures unobservable factors affecting success, such as luck. Performance comparisons are then linked through a nonnegative directed network
\[
W=(w_{jk})_{j,k=1}^n,
\qquad
w_{jk}\geq 0,
\qquad
w_{jj}=0.
\]
Here $w_{jk}$ measures how the performance of position $k$ enters the benchmark faced by position $j$, and $\beta\geq 0$ scales the overall strength of these benchmark interactions. We refer to $W$ as the \emph{competition structure}.

The competition structure includes the basic case in which all $w_{jk}=\frac{1}{n-1}$, so that an individual stands out when their position's performance exceeds the average of the peer group, multiplied by $\beta$. It also allows for more flexible comparison patterns that better capture real-world settings. In the sales example above, different customer segments may be evaluated differently, so salespeople are mainly compared with others facing similar customers and sales targets rather than with the whole firm. For instance, salespeople serving small clients may be benchmarked mainly against those serving other small clients or medium-sized clients.

Based on Equation \eqref{eq:success_condition}, define the standardized success margin
\begin{equation}
z_j^{(m)}
:=
\frac{x_j^{(m)}-\beta\sum_{k=1}^n w_{jk}x_k^{(m)}}{\sigma}.
\label{eq:z_def}
\end{equation}
Then the probability of success at position $j$ is
\[
\Pr(\text{success at }j)=\Phi\bigl(z_j^{(m)}\bigr),
\]
where $\Phi(\cdot)$ denotes the standard normal cdf.

Let
\begin{equation}
E_j^{(m)}:=e_{m^{-1}(j)},
\qquad j\in J,
\label{eq:position_effort}
\end{equation}
denote the effort exerted at position $j$. Each successful agent receives a common reward $U>0$ (success is not mutually exclusive across agents). Therefore, if agent $i$ is assigned to position $j=m(i)$, her expected payoff is
\begin{equation}
EU_i\bigl(E^{(m)};\omega,m\bigr)
=
U\Phi\bigl(z_j^{(m)}\bigr)-vE_j^{(m)},
\label{eq:agent_payoff}
\end{equation}
where $v>0$ is the marginal cost of effort.

Given $(m,\omega)$, a Stage-2 equilibrium is an effort profile
\[
E^{*(m)}(\omega)=\big(E_1^{*(m)}(\omega),\dots,E_n^{*(m)}(\omega)\big)
\]
such that each agent's effort is optimal given the efforts of all other agents.

\subsection{Stage 1: planner's problem}
\label{subsec:planner_problem}

The planner cares about aggregate true output rather than the internal evaluation score.\footnote{While we refer to this as \say{true output}, it is more generally intended to capture what the planner truly cares about, and need not coincide with output in the usual sense.} At position $j$, true output is given by
\begin{equation}
y_j^{(m)}:=sE_j^{(m)}+(1-s)\bigl(\theta_j^{(m)}+\alpha\bigr),
\qquad
s\in(0,1].
\label{eq:true_output}
\end{equation}
Thus, if $\omega=s$, the evaluation rule coincides with true output, so agents are evaluated according to the output they produce. When $\omega\neq s$, the planner deliberately tilts evaluation away from true output in order to affect incentives. Equivalently,
\begin{equation}
x_j^{(m)}
=
y_j^{(m)}
+
(\omega-s)\Bigl(E_j^{(m)}-\bigl(\theta_j^{(m)}+\alpha\bigr)\Bigr).
\label{eq:evaluation_decomposition}
\end{equation}

\paragraph{Empirical interpretation of $s$ and $\omega$.} The distinction between $s$ and $\omega$ is empirically important. The parameter $s$ describes the true production technology. It can be seen as the weight of effort in the output that the organization ultimately values. By contrast, $\omega$ describes the evaluation rule used to allocate rewards, promotions, recognition, or other career benefits. In applications, $s$ can be estimated from an output equation relating realized performance to effort and non-discretionary advantage, whereas $\omega$ can be inferred from scorecards, evaluation rubrics, supervisor weights, or policy reforms that change how much effort-based metrics enter performance assessment. Hence deviations of $\omega$ from $s$ are not merely theoretical. They correspond to observable distortions between what the organization values and what it rewards.

The linearity of both true output and evaluation follows classical frameworks in personnel economics, such as \citet{prendergast1999provision}, where output and evaluation are referred to as the \say{objective measure} and the \say{subjective measure}. It is also common in principal-agent models to assume that output is a linear function of the relevant inputs.\footnote{See, for example, \cite{LazearRosen1981RankOrder,Sannikov2008}.} 

In our setting, once output is linearly determined by its inputs, a likewise linear evaluation rule allows the planner to reweight effort relative to the non-discretionary component. For example, a manager may treat contracts signed or revenue generated as a hard metric, but still wish to give extra credit to salespeople who visibly exert greater effort, like making more calls, attending more meetings, or visiting more clients, while holding realized output fixed. This corresponds to $s<\omega$, so that, in Equation \eqref{eq:evaluation_decomposition}, effort receives additional credit conditional on true output \(y_j^{(m)}\).\footnote{Equation \eqref{eq:evaluation_decomposition} implies that giving additional credit to effort necessarily reduces the
weight placed on the non-discretionary component, \(\theta_j^{(m)}+\alpha\). This follows from the
normalization that the weights between them sum to one. Relaxing this normalization does not fundamentally
change the model. To see this, suppose instead that the evaluation score were
\[
x_j^{(m)}=y_j^{(m)}+\omega_1E_j^{(m)}+\omega_2\big(\theta_j^{(m)}+\alpha\big),
\]
with arbitrary \(\omega_1\) and \(\omega_2\). Using
\(y_j^{(m)}=sE_j^{(m)}+(1-s)\big(\theta_j^{(m)}+\alpha\big)\), we can rewrite this as
\[
x_j^{(m)}
=
\Big(1+\frac{\omega_1}{s}\Big)y_j^{(m)}
+
\Big(\omega_2-\frac{1-s}{s}\omega_1\Big)\big(\theta_j^{(m)}+\alpha\big),
\]
or equivalently,
\[
x_j^{(m)}
=
\Big(1+\frac{\omega_2}{1-s}\Big)y_j^{(m)}
+
\Big(\omega_1-\frac{s}{1-s}\omega_2\Big)E_j^{(m)}.
\]
For \(s\in(0,1)\), the coefficient on \(E_j^{(m)}\) is positive if and only if the coefficient on
\(\theta_j^{(m)}+\alpha\) in the previous expression is negative:
\[
\omega_1-\frac{s}{1-s}\omega_2>0
\quad\Longleftrightarrow\quad
\omega_2-\frac{1-s}{s}\omega_1<0.
\]
Under this normalization, holding true output fixed, increasing the implicit credit given to one component necessarily reduces the implicit credit given to the other.}

The planner chooses the assignment $m$ and the evaluation rule $\omega$ to maximize aggregate true output, anticipating the Stage-2 equilibrium induced. Given that $E_j^{*(m)}(\omega)$ denotes the equilibrium effort at position $j$, the planner's objective is
\begin{equation}
\mathcal{W}(\omega,m)
:=
\sum_{j=1}^n y_j^{*(m)}(\omega)
=
s\sum_{j=1}^n E_j^{*(m)}(\omega)
+
(1-s)\sum_{j=1}^n \theta_j^{(m)}.
\label{eq:planner_objective}
\end{equation}
Hence the planner solves
\begin{equation}
\max_{\omega\in(0,1]}\max_{m\in\mathcal{M}} \mathcal{W}(\omega,m),
\label{eq:planner_problem_full}
\end{equation}
where $\mathcal{M}$ is the set of all bijections from $I$ to $J$.

The planner therefore faces a joint design problem: he assigns agents to positions and chooses $\omega$ so as to balance the output value of effective advantage against the incentive effects of evaluation. In the baseline model, we assume that agents simply take the positions assigned to them in Stage 2. In an extension, we consider the case in which agents can swap positions, in which case the planner faces an incentive-compatibility constraint.

In the analysis, we solve the model by backward induction. Given $(\omega,m)$, a Stage-2 equilibrium is a pure-strategy Nash equilibrium in effort choices. Anticipating the corresponding Stage-2 equilibrium, the planner then chooses $(\omega,m)$ to maximize the Stage-1 objective.

\section{Analysis}\label{analysis}
\subsection{Stage 2 equilibrium}
%===============================================================

Stage 2 is a contest game in a network setting, a class of models that has been studied extensively in the literature.\footnote{See, for example, \cite{konig2017networks,bozbay2018contest,faure2020perceived,matros2024contests}.} Common assumptions in this literature limit strategic interdependence and restrict the benefits and costs of effort so as to ensure a well-defined interior equilibrium. The two assumptions below serve these purposes here.

\begin{assumption}[Network stability]\label{ass:network-stability}
\[
0\le \beta < \frac{1}{\rho(W)}.
\]
\end{assumption}

where $\rho(W)$ denotes the spectral radius of $W$. This is a standard assumption in the network strategic interaction literature, and ensures that agents are not overly sensitive to changes in others’ effort choices, so the equilibrium does not become explosive.

The next assumption defines the active-contest region. It ensures that, for all
evaluation weights in this region, the candidate interior effort profile is positive and no agent strictly prefers the zero-effort corner. For this assumption and for later use, define
\[
\underline\omega:=\frac{v\sigma}{U\phi(0)}
=\frac{v\sigma\sqrt{2\pi}}{U},
\]
and, for $\omega>\underline\omega$,
\[
\bar z(\omega):=
\sqrt{2\log\!\left(\frac{\omega}{\underline\omega}\right)}.
\]
Equivalently, $\bar z(\omega)$ is the positive root of
\[
\phi(\bar z(\omega))=\frac{v\sigma}{U\omega}.\footnote{Since $\phi(z)=(1/\sqrt{2\pi})\exp(-z^2/2)$ and
$\underline\omega=v\sigma\sqrt{2\pi}/U$, the equation
$\phi(z)=v\sigma/(U\omega)$ is equivalent to
$\exp(-z^2/2)=\underline\omega/\omega$, hence
$z^2=2\log(\omega/\underline\omega)$. The positive root is therefore
$z=\sqrt{2\log(\omega/\underline\omega)}$.}
\]
Also define
\[
q:=(I-\beta W)^{-1}\mathbf 1.
\]
This is the \emph{Katz centrality} vector associated with the competition network $W$
and parameter $\beta$. It captures how strongly each position is amplified
through the network, taking into account both direct and indirect effects.

\begin{assumption}[Interior-equilibrium region]\label{ass:alpha-bound} There exists $\omega^I\in(\underline\omega,1]$ such that, for every
$\omega\in[\omega^I,1]$, every assignment $m\in\mathcal M$, and every position
$j\in J$,
\[
0<
\bar z(\omega)q_j
-
\frac{1-\omega}{\sigma}
\left(
\sum_{k=1}^n M_{jk}a_k^{(m)}+\alpha
\right)
<
2\bar z(\omega).
\]

\end{assumption}

The first inequality ensures strictly positive candidate effort. The
second inequality rules out cases in which the agent would prefer to deviate to zero effort. Assumption~\ref{ass:alpha-bound} equivalently requires the non-discretionary
component at each position, $\sum_{k=1}^n M_{jk}a_k^{(m)}+\alpha$, to lie in an interval. 

The assumption is technical, but it also has a natural interpretation. The paper
focuses on active-contest environments in which all agents remain meaningfully
engaged in relative performance evaluation. It therefore excludes cases in which
an agent's non-discretionary advantage is either so high that effort is not
needed, or so low that the agent gives up and chooses zero effort. Intuitively, one can also think that a manager would not tolerate zero effort and would typically intervene before such a corner outcome is reached. In any case, the corner cases are not relevant to our comparative statics. As it is natural to believe that, in most organizations, people rarely exert no effort at all, we choose not to focus on those cases.\footnote{The assumption would also look less restrictive if the universal $\alpha$ were replaced by an agent-specific term $\alpha_i$ or a position-specific term $\alpha_j$. Such variants would not change the main results of the model.}

\begin{lemma}[No effort below $\underline\omega$]\label{lem:collapse-threshold}
For any assignment $m$, if $\omega\le \underline\omega$, then the unique Stage-2
pure-strategy Nash equilibrium is
\[
E^{*(m)}(\omega)\equiv 0.
\]
\end{lemma}

\proofinapp{app:proof-lemma1}

Intuitively, if $\omega\le \underline\omega$, each agent always weakly prefers to
reduce effort, because the marginal benefit of increasing effort is always no
larger than the marginal cost $v$. For $\omega\in(\underline\omega,\omega^I)$, the first-order condition admits a
positive root, but the closed-form candidate need not be a Nash equilibrium. In
such a region, some agents may prefer the zero-effort corner, and pure equilibria,
when they exist, may involve corner effort choices. The
main characterization is below for $\omega\in[\omega^I,1]$, which throughout the paper is called \emph{active-contest region}, where
Assumption~\ref{ass:alpha-bound} guarantees that the closed-form candidate is an
all-active Stage-2 equilibrium and that it is the unique pure-strategy equilibrium
in which all agents exert strictly positive effort.

\begin{proposition}[Second-stage equilibrium]
\label{prop:stage2-omega} Maintain Assumptions~\ref{ass:network-stability}
and \ref{ass:alpha-bound}. Given $(\omega,m)$ with $\omega\ge\omega^I$, the
Stage~2 game admits the following all-active pure-strategy Nash equilibrium: 
\begin{equation*}
x^{*(m)}(\omega)=\sigma \bar z(\omega)\,q, \qquad E^{*(m)}(\omega)=\frac{%
\sigma}{\omega}\bar z(\omega)\,q -\frac{1-\omega}{\omega}\big(%
\theta^{(m)}+\alpha\mathbf{1}\big).
\end{equation*}
Moreover, this is the unique pure-strategy equilibrium in which all agents
exert positive effort. In this equilibrium, 
\begin{equation*}
z^{(m)}_j=\bar z(\omega) \qquad\text{and}\qquad \Pr(\text{success at }%
j)=\Phi(\bar z(\omega))
\end{equation*}
for all $j$.
\end{proposition}
\proofinapp{app:proof-stage2-omega}

\paragraph{Empirical implications of equilibrium effort.} The closed-form effort expression in Proposition~\ref{prop:stage2-omega} generates testable implications for individual effort. In the active-contest region, \[ E_j^{*(m)}(\omega) = \frac{\sigma}{\omega}\bar z(\omega)q_j - \frac{1-\omega}{\omega}\bigl(\theta_j^{(m)}+\alpha\bigr). \] Therefore, holding effective advantage fixed, effort is increasing in the Katz centrality of the position in the competition network: \[ \frac{\partial E_j^{*(m)}(\omega)}{\partial q_j} = \frac{\sigma}{\omega}\bar z(\omega)>0. \] Workers in positions that are more central in the benchmarking structure should therefore exert more observable effort, such as calls, visits, hours, client contacts, applications processed, or other measures of work intensity. The same expression implies that, holding network centrality fixed, effort is decreasing in effective advantage: \[ \frac{\partial E_j^{*(m)}(\omega)}{\partial \theta_j^{(m)}} = -\frac{1-\omega}{\omega}\le 0. \] Thus agents with greater non-discretionary advantage should exert less effort, because the advantage makes success easier. Moreover, \[ \frac{\partial^2 E_j^{*(m)}(\omega)} {\partial \omega\,\partial \theta_j^{(m)}} = \frac{1}{\omega^2}>0, \] so increasing the effort weight in evaluation weakens the negative effect of advantage on effort. This gives a difference-in-differences implication. After a reform that raises the weight placed on effort in evaluation, effort should rise relatively more among agents with high effective advantage.

This prediction is related to empirical evidence that relative performance information and peer comparisons affect effort and productivity. For example, \citet{MasMoretti2009Peers} show that peer productivity affects worker effort in a supermarket setting, especially when coworkers can observe one another. \citet{BlanesVidalNossol2011} find that providing workers with feedback about their relative position in the productivity distribution increases productivity, while \citet{AzmatIriberri2010} find similar effects of relative performance feedback in an educational setting. Related evidence that relative incentives can alter effort, especially when individual effort imposes externalities on others, is provided by \citet{BandieraBarankayRasul2005}.

\vspace{15pt}

Proposition~\ref{prop:stage2-omega} implies that an agent’s equilibrium effort is increasing in her Katz centrality in the active-contest region. This is consistent with the classical literature on effort and strategic interaction on networks.\footnote{For example, \cite{ballester2006s} and \cite{konig2017networks}}. Another noteworthy feature is that the equilibrium effort need not be monotone in $\omega$. Differentiating yields
\begin{equation}\label{eq:effort-derivative-omega}
\frac{\partial E^{*(m)}(\omega)}{\partial \omega}
=
\frac{\sigma q}{\omega^2}\left(\frac{1-\bar z^2(\omega)}{\bar z(\omega)}\right)
+\frac{\theta^{(m)}+\alpha\mathbf 1}{\omega^2}.
\end{equation}
Its sign is generally ambiguous, since the sign of $1-\bar z^2(\omega)$ is ambiguous. Intuitively, when $\omega$ is small, performance depends more heavily on effective advantage, so less-advantageous agents face a greater disadvantage and may exert more effort to compensate. By contrast, when $\omega$ is large, performance depends more directly on effort, which may weaken this compensating motive. 

This can also be seen from the role of $\theta^{(m)}$ in Equation~\ref{eq:effort-derivative-omega}. The derivative may be negative when effective advantage is low, and likely positive when effective advantage is high. Intuitively, for less-advantaged agents, a lower $\omega$ means that their disadvantage matters more, which strengthens their incentive to make up for it through effort. Conversely, for more-advantaged agents, this force is weaker.

It is therefore not generally optimal for the planner to make $\omega$ as large as possible. In the next section, we further show that it is also generically not optimal for the planner to set $\omega=s$.

%===============================================================
\subsection{Planner's optimal evaluation and assignment}
%===============================================================

In the previous subsection, we took $(\omega,m)$ as given. We now turn to the planner's problem and analyze the optimal choice of $(\omega,m)$, assuming that the planner anticipates the Stage-2 outcome and that agents cannot change positions after the assignment is chosen. Since the unique Stage-2 equilibrium is guaranteed only when $\omega\in[\omega^I,1]$, we restrict the planner's choice of
$\omega$ to this interval. Thus, the planner's problem becomes
\begin{equation}
\max_{\omega\in[\omega^I,1]}\max_{m\in\mathcal{M}} \mathcal{W}(\omega,m).
\label{eq:planner_problem_active}
\end{equation}
We justify this restriction later in this section by discussing what the
planner's payoff might be when the planner chooses $\omega\in[0,\omega^I)$. Recall $q:=(I-\beta W)^{-1}\mathbf 1$ and write
$Q:=\mathbf 1^\top q$, i.e., the sum of Katz centrality indices. Using
\[
E^{*(m)}(\omega)
=
\frac{\sigma}{\omega}\bar z(\omega)\,q
-
\frac{1-\omega}{\omega}\big(\theta^{(m)}+\alpha\mathbf 1\big),
\]
we obtain
\[
\sum_j E^{*(m)}_j(\omega)
=
\frac{\sigma}{\omega}\bar z(\omega)\,Q
-
\frac{1-\omega}{\omega}
\big(\mathbf 1^\top\theta^{(m)}+n\alpha\big).
\]
Plugging this into Equation~\eqref{eq:planner_objective} yields
\begin{equation}\label{eq:planner-reduced-omega}
\mathcal W(\omega,m)
=
s\,\frac{\sigma}{\omega}\bar z(\omega)\,Q
-
s\,\frac{1-\omega}{\omega}\,n\alpha
+
\Big(1-\frac{s}{\omega}\Big)\mathbf 1^\top\theta^{(m)}.
\end{equation}
Since $\mathbf 1^\top\theta^{(m)}=b^\top a^{(m)}$ with
$b=M^\top\mathbf 1$, the assignment-relevant term is
\[
\Big(1-\frac{s}{\omega}\Big)b^\top a^{(m)}.
\]
Hence, conditional on $\omega\in[\omega^I,1]$, the planner's assignment problem
reduces to maximizing a one-dimensional criterion. The lemma below makes this
reduction explicit.

\begin{lemma}[One-dimensional reduction of the assignment problem]\label{lem:reduction-b}
Maintain Assumptions~\ref{ass:network-stability} and~\ref{ass:alpha-bound}, and
restrict the planner's choice of $\omega$ to $[\omega^I,1]$. The planner's objective
can be written as
\[
\mathcal W(\omega,m)
=
A(\omega)
+
\Big(1-\frac{s}{\omega}\Big)b^\top a^{(m)},
\]
where $A(\omega)$ does not depend on $m$. Consequently, conditional on $\omega$,
the planner chooses $m$ to maximize $b^\top a^{(m)}$ when $\omega>s$, to minimize
$b^\top a^{(m)}$ when $\omega<s$, and is indifferent over assignments when
$\omega=s$.
\end{lemma}

\begin{proof}
This follows directly from Equation~\eqref{eq:planner-reduced-omega} and
$\mathbf 1^\top\theta^{(m)}=b^\top a^{(m)}$.
\end{proof}
Note that Lemma~\ref{lem:reduction-b} states properties of the optimal assignment, rather than characterizing the planner’s assignment mechanism. Note also that the planner does not distort evaluation for its own sake. The distortion is useful because it changes the effort response induced by the relative benchmark. Proposition~\ref{prop:opt-assign-given-omega} provides a formal characterization.

\begin{proposition}[Optimal assignment]\label{prop:opt-assign-given-omega}
Maintain Assumptions~\ref{ass:network-stability} and~\ref{ass:alpha-bound}, and
restrict the planner's choice of $\omega$ to $[\omega^I,1]$. Fix
$\omega\in[\omega^I,1]$. Let $(i(1),\dots,i(n))$ be any ordering of agents such that
\begin{equation}\label{eq:agents-sorted-by-ability}
a_{i(1)}\le \cdots \le a_{i(n)},
\end{equation}
and let $(j(1),\dots,j(n))$ be any ordering of positions such that
\begin{equation}\label{eq:positions-sorted-by-b}
b_{j(1)}\le \cdots \le b_{j(n)}.
\end{equation}
Then an assignment is planner-optimal conditional on $\omega$ if and only if it has
the following form, up to arbitrary permutations within tied values of $\{a_i\}$
and/or tied values of $\{b_j\}$:
\begin{equation}\label{eq:opt-assignment-given-omega}
m^\star(i(k))=
\begin{cases}
j(k), & \text{if }\omega>s,\\
j(n+1-k), & \text{if }\omega<s,\\
\text{any bijection from $I$ to $J$,} & \text{if }\omega=s,
\end{cases}
\qquad k=1,\dots,n.
\end{equation}
\end{proposition}

\proofinapp{app:proof-opt-assign-given-omega}

Proposition~\ref{prop:opt-assign-given-omega} has a simple intuition. Effective advantage raises output directly, but it also depresses effort by making winning easier. When $\omega>s$, the direct effect dominates, so the planner prefers an assignment that maximizes total effective advantage. We call this a \emph{positive assortative assignment}. When $\omega<s$, the incentive effect dominates, so the planner instead prefers an assignment that minimizes total effective advantage. We call this a \emph{negative assortative assignment}. Note that it is possible that $s<\omega^I$. In that case, the optimal assignment is always positive assortative, given we restrict the planner's choice $\omega\in[\omega^I,1]$.

\paragraph{Empirical implication for assignment sorting.} Proposition~\ref{prop:opt-assign-given-omega} implies a sharp prediction for the sorting of agents across positions. Let ability be measured by $a_i$ and let position spillover centrality be measured by $b_j$. When $\omega>s$, the planner assigns high-ability agents to high-$b_j$ positions, so the rank correlation between $a_i$ and $b_{m(i)}$ should be positive. When $\omega<s$, the planner assigns high-ability agents to low-$b_j$ positions, so the same rank correlation should be negative. Thus the sign of assortative matching should vary with the gap between the evaluation rule and the true production technology. This implication is distinct from sorting on own-position advantage. Since $b_j=\sum_r M_{rj}$ measures the total spillover contribution of position $j$, positive assortative assignment means sorting by a position's aggregate contribution to effective advantage, not necessarily by its diagonal term $M_{jj}$. Empirically, this distinction matters when some positions generate large spillovers for others even if their own direct advantage is not especially large.

The prediction that effort responds to the weight placed on measured performance is consistent with evidence from personnel settings. \citet{Lazear2000} shows that the introduction of piece rates substantially increased productivity in a firm, while \citet{BandieraBarankayRasul2007} show that managerial performance pay changed how managers allocated attention across workers. More broadly, \citet{AshrafBandieraJack2014} show that the design of financial and non-financial rewards affects agents' performance in a task with social spillovers.

\vspace{15pt} 

Define \(B^+:=b^\top a^{(m^+)}\) and \(B^-:=b^\top a^{(m^-)}\), where \(m^+\) and \(m^-\) denote positive- and negative-assortative assignments, respectively. Throughout this paper, we assume \(B^+>B^-\), since otherwise the assignment problem would be trivial. Proposition~\ref{prop:opt-assign-given-omega} implies that the planner selects the
assignment attaining $B^+$ when $\omega\ge s$ and the assignment attaining $B^-$
when $\omega\le s$. Accordingly, over the active-contest region
$\omega\in[\omega^I,1]$, the planner's value on each relevant branch is
\begin{equation}\label{eq:V-plus}
V^+(\omega):=
s\,\frac{\sigma}{\omega}\bar z(\omega)\,Q
-\;s\,\frac{1-\omega}{\omega}\,n\alpha
+\Big(1-\frac{s}{\omega}\Big)B^+,
\qquad \omega\in[\max\{s,\omega^I\},1],
\end{equation}
and, if $\omega^I\le s$,
\begin{equation}\label{eq:V-minus}
V^-(\omega):=
s\,\frac{\sigma}{\omega}\bar z(\omega)\,Q
-\;s\,\frac{1-\omega}{\omega}\,n\alpha
+\Big(1-\frac{s}{\omega}\Big)B^-,
\qquad \omega\in[\omega^I,s].
\end{equation}
If $\omega^I>s$, the negative-assortative branch is empty.

Each branch admits a closed-form stationary point. Recall that $\underline\omega=\frac{v\sigma\sqrt{2\pi}}{U}$.
For any $B\in\{B^+,B^-\}$, define
\begin{equation}\label{eq:K-z-omegahat}
K(B):=\frac{B+n\alpha}{\sigma Q},
\qquad
z(B):=\frac{K(B)+\sqrt{K(B)^2+4}}{2},
\qquad
\widehat\omega(B):=\underline\omega\exp\!\Big(\frac{z(B)^2}{2}\Big).
\end{equation}
It can be shown that for a given value of $B$, the candidate stationary point is summarized by
the pair $(z(B),\widehat\omega(B))$. 

Define the feasible positive-assortative candidate by
\begin{equation}\label{eq:omega-plus-star}
\omega_+^\ast
:=
\min\left\{1,\max\left\{\max\{s,\omega^I\},\widehat\omega(B^+)\right\}\right\}.
\end{equation}

Define the feasible negative-assortative candidate by
\begin{equation}\label{eq:omega-minus-star}
\omega_-^\ast
:=
\min\left\{s,\max\left\{\omega^I,\widehat\omega(B^-)\right\}\right\}.
\end{equation}

Since $B^+\ge B^-$, we have $\widehat\omega(B^+)\ge \widehat\omega(B^-)$. Moreover,
by definition,
\[
\omega_+^\ast
\ge
\max\{s,\omega^I\},
\]
since $s,\omega^I\le 1$. And
\[
\omega_-^\ast
\le s.
\]
Therefore,
\[
\omega_+^\ast\ge s\ge \omega_-^\ast.
\]
Hence the feasible positive-assortative candidate is weakly larger than the
feasible negative-assortative candidate. Additionally, let $m^+$ denote a positive assortative assignment and $m^-$ denote a negative assortative assignment. Write a policy as $(\omega,m)$. The following lemma and proposition then characterize the planner's optimal policy in different cases.

\begin{lemma}[Optimal policy when $s<\omega^I$]\label{lem:positive-only}
Maintain Assumptions~\ref{ass:network-stability} and~\ref{ass:alpha-bound}, and
restrict the planner's choice of $\omega$ to $[\omega^I,1]$. If $s<\omega^I$,
then the negative-assortative branch is empty, and the policy $(\omega_+^\ast,m^+)$ is optimal for the planner.
\end{lemma}
\proofinapp{app:proof-lemma3}

We treat the lemma above as characterizing a special case in which a unique closed-form equilibrium possibly arises only when the evaluation weight on effort is higher than the weight of effort in true output. This could mean that the weight of effort in true output is too low. The proposition below gives the paper's main result. It characterizes the more general case in which the planner chooses $\omega$ from a wider range.

\begin{proposition}[Optimal policy when $\omega^I\le s$]\label{prop:omega-cutoff}
Maintain Assumptions~\ref{ass:network-stability} and~\ref{ass:alpha-bound}, and
restrict the planner's choice of $\omega$ to $[\omega^I,1]$. Suppose
$\omega^I\le s$.

\emph{Case 1.} If $1<\widehat\omega(B^-)\le \widehat\omega(B^+)$, then the policy
$(\omega_+^\ast,m^+)$ is optimal for the planner.

\emph{Case 2.} If $\widehat\omega(B^+)<\omega^I\le s$, then:
\begin{enumerate}[topsep=2pt,itemsep=2pt,parsep=0pt,partopsep=0pt]
\item If $s=\omega^I$, every policy $(\omega^I,m)$ is optimal, regardless of the assignment $m$.
\item If $s>\omega^I$, the policy
$(\omega_-^\ast,m^-)$
is optimal for the planner.
\end{enumerate}

\emph{Case 3.} If $\widehat\omega(B^-)\le 1$ and $\omega^I\le \widehat\omega(B^+)$,
then there exists a cutoff $s^\star$ such that:
\begin{enumerate}[topsep=2pt,itemsep=2pt,parsep=0pt,partopsep=0pt]
\item If $s<s^\star$, the policy
$(\omega_+^\ast,m^+)$
is optimal for the planner.
\item If $s>s^\star$, the policy
$(\omega_-^\ast,m^-)$
is optimal for the planner.
\item If $s=s^\star$, the planner is indifferent between the two policies.
\end{enumerate}

The cutoff $s^\star$ is defined as follows. Let
\[
\widetilde T^-
:=
\begin{cases}
n\alpha+\dfrac{\sigma Q}{\widehat\omega(B^-)z(B^-)},
& \text{if } \omega^I\le \widehat\omega(B^-),\\[1.2em]
n\alpha+\dfrac{\sigma Q\bar z(\omega^I)-B^- -n\alpha}{\omega^I},
& \text{if } \widehat\omega(B^-)<\omega^I,
\end{cases}
\]
and
\[
\widetilde T^+
:=
\begin{cases}
n\alpha+\dfrac{\sigma Q}{\widehat\omega(B^+)z(B^+)},
& \text{if } \widehat\omega(B^+)\le 1,\\[1.2em]
\sigma Q\bar z(1)-B^+,
& \text{if } \widehat\omega(B^+)>1.
\end{cases}
\]

Then define
\[
s^\star
:=
\dfrac{B^+-B^-}{\widetilde T^- - \widetilde T^+}.
\]
In particular, $\widetilde T^- - \widetilde T^+>0$, given $B^+>B^-$ as assumed.
\end{proposition}

\proofinapp{app:proof-omega-cutoff}

Proposition~\ref{prop:omega-cutoff} highlights that the planner generally does not evaluate performance by true output alone, i.e. set $\omega=s$. Instead, the optimal evaluation may overweight effort, with $\omega>s$, or overweight effective advantage, with $\omega<s$, depending on the relative importance of effort in true output. Interestingly, the result suggests that the planner has an incentive to choose a higher evaluation weight, $\omega_+^\ast$, when the weight of effort in true output is relatively low, i.e. $s<s^\star$, and vice versa. The special cases, such as Lemma~\ref{lem:positive-only} and Case 1 in Proposition~\ref{prop:omega-cutoff}, do not contradict this logic. The case in Lemma~\ref{lem:positive-only}, where $s<\omega^I$, and Case 1 in Proposition~\ref{prop:omega-cutoff}, where $1<\widehat\omega(B^-)\le \widehat\omega(B^+)$, can both be understood as cases in which $s$ is too low relative to the planner's feasible and optimal candidate value of $\omega$, and hence the higher value $\omega_+^\ast$ is chosen.

The intuition for this result must be understood jointly with the planner's assignment choice. Equation~\eqref{eq:effort-derivative-omega} suggests that, roughly speaking, a higher $\omega$ tends to encourage effort by agents with high effective advantage, whereas a lower $\omega$ tends to encourage effort by agents with low effective advantage. When effective advantage is relatively important for output (the $s<s^*$ case), the planner wants to maximize total advantage and therefore chooses a positive assortative assignment. In that case, a higher $\omega$ helps raise total effort by strengthening incentives for agents who are already rich in effective advantage. By contrast, when effort is relatively important for output (the $s>s^\star$ case), the planner no longer has a sufficiently strong incentive to increase total effective advantage. It therefore chooses a negative assortative assignment, which creates relatively disadvantaged agents and allows the planner to use a lower $\omega$ to induce them to exert more effort in order to make up for their disadvantage.

While Lemma~\ref{lem:positive-only} and Proposition~\ref{prop:omega-cutoff} fully characterize the planner's optimal policies, they involve several special cases, which may distract from the main idea. We therefore focus on the most natural case, in which $\omega^I\le \hat \omega(B^-)\le \hat \omega(B^+)\le 1$. In this case, the planner's candidate evaluation weights are interior. The following assumption guarantees this ordering.  

\begin{assumption}\label{ass:interior-candidate-ordering} Assume
\[
\sigma Q\left[
\bar z(\omega^I)-\frac{1}{\bar z(\omega^I)}
\right]
\le
B^-+n\alpha
<
B^+ + n\alpha
\le
\sigma Q\left[
\bar z(1)-\frac{1}{\bar z(1)}
\right],
\]
where
\[
\bar z(\omega):=\sqrt{2\log\left(\frac{\omega}{\underline\omega}\right)}.
\]
\end{assumption}

Assumption~\ref{ass:interior-candidate-ordering} requires the aggregate non-discretionary component to be neither too low nor too high relative to the incentive environment. Intuitively, it means that effective advantage is important enough to affect the planner's evaluation choice, but not so dominant that the planner always wants to put the maximal feasible weight on effort. Technically, it implies $\omega^I\le \hat \omega(B^-)\le \hat \omega(B^+)\le 1$, as shown in Appendix~\ref{app:below-omegaI}. Appendix~\ref{app:below-omegaI} also discusses when Proposition~\ref{prop:omega-cutoff} continues to hold under Assumption~\ref{ass:interior-candidate-ordering}, even if the planner can choose $\omega$ below $\omega^I$, as long as an equilibrium exists for those values. Unless explicitly stated otherwise, the rest of the paper maintains
Assumptions~\ref{ass:network-stability}, \ref{ass:alpha-bound}, and
\ref{ass:interior-candidate-ordering}.

Given Assumptions~\ref{ass:network-stability}--\ref{ass:interior-candidate-ordering}, we can clearly analyze how the planner's cutoff strategy depends on both the spillover structure $M$ and the competition structure $W$.

\begin{corollary}[Comparative statics of the cutoff]\label{cor:sstar-comparative}
Maintain Assumptions~\ref{ass:network-stability}, \ref{ass:alpha-bound}, and
\ref{ass:interior-candidate-ordering}. Suppose $\underline\omega e^{1/2}\le \omega^I\le s$. Then the cutoff $s^\star$ is weakly decreasing in $Q$. Moreover, holding $B^-$ fixed, $s^\star$ is weakly increasing in $B^+$.
\end{corollary}

\proofinapp{app:proof-corollary}

This follows the same intuition as Proposition~\ref{prop:omega-cutoff}. Recall that $Q$ is the sum of Katz centralities induced by $W$, and a higher $Q$ corresponds to more intense competition. When competition is more intense, the planner is more likely to choose a negative assortative assignment together with a lower $\omega$, because a denser competition structure strengthens disadvantaged agents' incentives in this case to exert effort to make up for their disadvantage. Also, holding $B^-$ fixed, a higher $B^+$ raises the cutoff, because for a given type profile, this corresponds to a more uneven spillover structure and makes it more attractive for the planner to maximize total effective advantage, even at the cost of inducing less effort.

Overall, the planner faces a tradeoff between increasing total effective advantage and inducing more effort, and the optimal policy depends on the structures of competition and advantage spillovers. The example below helps illustrate this tradeoff more clearly.

\paragraph{Empirical implications of the cutoff.} Corollary~\ref{cor:sstar-comparative} gives cross-environment predictions. Since $s^\star$ is decreasing in $Q$, organizations or teams with more intense benchmark competition should be more likely to choose a low evaluation weight and a negative assortative assignment, all else equal. In such environments, assignment is used less to maximize total effective advantage and more to create incentives for disadvantaged agents to exert effort. The corollary also implies that, holding $B^-$ fixed, a larger $B^+$ raises $s^\star$. Hence environments with more unequal spillover opportunities should be more likely to use positive assortative assignment. If placing a high-type agent in the most central position generates a large increase in total effective advantage, the planner is more willing to tolerate weaker effort incentives in order to exploit those spillovers.

The sorting implication is more specific to the model. Existing field evidence nevertheless shows that incentive schemes can affect both productivity and the composition of teams. In particular, \citet{BandieraBarankayRasul2013} show that team incentives affect both effort and endogenous team composition. The model's prediction that the sign of assignment sorting should change with the gap between the evaluation rule and the true production technology is, to our knowledge, a novel implication that calls for further empirical work.

\begin{example}
Consider three agents with abilities
\[
(a_1,a_2,a_3)=(0,1,2),
\]
and fix
\[
\alpha=0,\qquad \sigma=1,\qquad U=10,\qquad v=0.2,\qquad \beta=0.25,\qquad s=0.25.
\]
Then
\[
\underline\omega=\frac{v\sigma\sqrt{2\pi}}{U}\approx 0.050.
\]

We use a star network (\(W^S\)), in which positions 1 and 3 compete only with position 2, to represent the less dense competition structure, and a complete network (\(W^C\)), in which each agent competes with both of the others, to represent the dense competition structure. For the spillover structure, we use a star-shaped matrix (\(M^S\)), in which one central position receives and transmits much stronger spillovers than the others, to represent the relatively uneven case, and a symmetric matrix with similar entries across positions (\(M^C\)) to represent the relatively even case. For all cases discussed below, one can show that there exists some $\omega^I$
such that Assumptions~\ref{ass:network-stability}--\ref{ass:interior-candidate-ordering}
hold. For example, take $\omega^I=0.25$ in case (i) and $\omega^I=0.15$
in case (ii).

\medskip
\noindent
\textbf{(i) Uneven and sparse interaction.}
Let
\[
W^{S}=
\begin{pmatrix}
0&1&0\\
1&0&1\\
0&1&0
\end{pmatrix},
\qquad
M^{S}=I+W^{S}.
\]
In this case, it can be computed that $s^{*}\approx 0.384>s$. Hence the planner chooses the positive assortative policy. Under
\[
a^{(m^+)}=(1,2,0),
\qquad
\widehat\omega(B^+)\approx 0.568,
\]
the effective-advantage and effort profiles are
\[
\theta^{(m^+)}=M^{S}a^{(m^+)}=
\begin{pmatrix}
3\\
3\\
2
\end{pmatrix},
\qquad
E^{*(m^+)}\bigl(\widehat\omega(B^+)\bigr)\approx
\begin{pmatrix}
3.26\\
4.37\\
4.02
\end{pmatrix}.
\]
In this case, the planner assortatively matches agents by type in order to generate higher effective talent. The resulting effective-advantage profile is quite uneven: position 3 is disadvantaged, with effective advantage 2, whereas the other two positions each have effective advantage 3. However, because effective advantage is relatively less important in evaluation, the effort profile does not show position 3 exerting substantially more effort to catch up. In fact, the agent in position 2 exerts the most effort, reflecting the centrality of position 2 in the competition network.

\medskip
\noindent
\textbf{(ii) More even and denser interaction.}
Let
\[
W^{C}=
\begin{pmatrix}
0&1&1\\
1&0&1\\
1&1&0
\end{pmatrix},
\qquad
M^{C}=
\begin{pmatrix}
1&0.6&0.5\\
0.6&1&0.5\\
0.5&0.5&1
\end{pmatrix}.
\]
In this case, it can be computed that $s^{*}\approx 0.193<s$. Hence the planner chooses the negative assortative policy. Under
\[
a^{(m^-)}=(0,1,2),
\qquad
\widehat\omega(B^-)\approx 0.189,
\]
the effective-advantage and effort profiles are
\[
\theta^{(m^-)}=M^{C}a^{(m^-)}=
\begin{pmatrix}
1.6\\
2.0\\
2.5
\end{pmatrix},
\qquad
E^{*(m^-)}\bigl(\widehat\omega(B^-)\bigr)\approx
\begin{pmatrix}
10.37\\
8.66\\
6.52
\end{pmatrix}.
\]
In this case, the planner uses a negatively assortative assignment, so the gaps between advantaged and disadvantaged agents are not large. However, agents’ effective advantages are lower overall. Importantly, effective advantage is made much more important in evaluation, which induces substantially greater effort, especially from the still least-advantaged agent at position 1.

\end{example}

%===============================================================
\section{Pairwise Stable Assignment and Stability Loss}
%===============================================================
\label{sec:pairwise}
So far, the planner's choice $(\omega,m)$ has been treated as fully enforceable. In many organizations, however, agents may be able to mutually agree to swap tasks after the assignment is made. When such swaps are feasible, the planner's chosen assignment must be incentive-compatible in the sense that no two agents can both strictly improve their continuation values by exchanging positions, which is known as \emph{pairwise stability}. This section extends the model to study the planner's problem when pairwise stability must be taken into account.

\subsection{Continuation values and pairwise stability}

We start by deriving the condition under which a pair of agents would agree to swap their positions, and define pairwise stability on this basis. Suppose Assumptions~\ref{ass:network-stability}--\ref{ass:interior-candidate-ordering}
hold. Define agent $i$'s continuation value as her Stage-2 equilibrium payoff
\[
V_i(\omega,m):=\mathbb{E}U_i\!\big(E^{*(m)}(\omega);\omega,m\big).
\]
By Proposition~\ref{prop:stage2-omega}, if $j=m(i)$ then
\begin{equation}\label{eq:Vi-closed-omega}
V_i(\omega,m)
=
U\Phi(\bar z(\omega))
-v\sigma\frac{\bar z(\omega)}{\omega}\,q_{j}
+v\frac{1-\omega}{\omega}\big(\theta^{(m)}_{j}+\alpha\big).
\end{equation}
The terms $U\Phi(\bar z(\omega))$ and $v\frac{1-\omega}{\omega}\alpha$ are common across positions, so they
drop out of swap comparisons.

\paragraph{Swap notation.}
For distinct agents $i\neq i'$, let $m^{ii'}$ denote the assignment obtained from $m$ by swapping
their positions. Write $j:=m(i)$ and $k:=m(i')$.

Fix an assignment $m$ and positions $(j,k)$. For $r\in\{j,k\}$, define the
contribution to his effective advantage $\theta^{(m)}_r$ that comes from all positions other than $\{j,k\}$ as
\[
R_r(m;j,k):=\sum_{\ell\notin\{j,k\}} M_{r\ell}\,a^{(m)}_\ell.
\]
Under $m$, we have
\[
\theta^{(m)}_j = M_{jj}a_i + M_{jk}a_{i'} + R_j(m;j,k),
\qquad
\theta^{(m)}_k = M_{kj}a_i + M_{kk}a_{i'} + R_k(m;j,k),
\]
while under the swapped assignment $m^{ii'}$,
\[
\theta^{(m^{ii'})}_j = M_{jj}a_{i'} + M_{jk}a_i + R_j(m;j,k),
\qquad
\theta^{(m^{ii'})}_k = M_{kj}a_{i'} + M_{kk}a_i + R_k(m;j,k).
\]
Therefore the swap gains, normalized by $v$, are
\begin{small}
\begin{align}
\label{eq:swap-gain-i}
\frac{1}{v}\big(V_i(\omega,m^{ii'})-V_i(\omega,m)\big)
&=
\sigma\frac{\bar z(\omega)}{\omega}\,(q_j-q_k)
+\frac{1-\omega}{\omega}\Big[\big(M_{kk}-M_{jj}\big)a_i+\big(M_{kj}-M_{jk}\big)a_{i'}+\big(R_k-R_j\big)\Big],
\\
\label{eq:swap-gain-ip}
\frac{1}{v}\big(V_{i'}(\omega,m^{ii'})-V_{i'}(\omega,m)\big)
&=
\sigma\frac{\bar z(\omega)}{\omega}\,(q_k-q_j)
+\frac{1-\omega}{\omega}\Big[\big(M_{jj}-M_{kk}\big)a_{i'}+\big(M_{jk}-M_{kj}\big)a_i+\big(R_j-R_k\big)\Big],
\end{align}
\end{small}
where $R_r$ abbreviates $R_r(m;j,k)$.

In some cases, agents can renegotiate roles with side transfers, so that a swap can be implemented whenever it strictly increases the pair's total payoff. This is the standard transferable-utility case. We adopt this transferable-utility view of swaps. Under this interpretation, a swap can be carried out whenever the two agents can jointly raise their total continuation value, which places the framework in the tradition of cooperative and coalition games with transferable utility.\footnote{See, for example, \cite{vonNeumannMorgenstern1944,Shapley1953,Gillies1959}. The same logic can also be viewed through the lens of network formation games with transfers. In classical network formation games, link formation is a joint action by a pair of agents that changes payoffs through network externalities, which is analogous to swapping positions in the present setting. \cite{BlochJackson2007} study pairwise link formation with transfers.} We think that transferable utility is plausible in organizational applications. In many workplace settings agents cannot literally make transfers, but they may exchange tasks, compensate through reciprocal help, or agree on informal workload trades, in quite incremental ways that simulate continuous transfers.

\begin{definition}[Profitable swap and stability]\label{def:stability}
Given $\omega$, the swap $(i,i')$ is \emph{profitable} under $m$ if
\[
V_i(\omega,m^{ii'})+V_{i'}(\omega,m^{ii'})
\;>\;
V_i(\omega,m)+V_{i'}(\omega,m).
\]
Given $\omega$, an assignment $m$ is \emph{stable} if no pair admits a profitable swap. 
$\mathcal M^{S}(\omega)$ denotes the set of stable assignments at $\omega$.
We call $(\omega,m)$ \emph{pairwise stable} if $m\in\mathcal M^{S}(\omega)$.
\end{definition}

Pairwise stable assignments are generally difficult to characterize. Equations~\eqref{eq:swap-gain-i} and~\eqref{eq:swap-gain-ip} imply that the pair's total continuation value depends on their specific positions in the network $M$, since the terms $\big(M_{kk}-M_{jj}\big)a_i$ and $\big(M_{jj}-M_{kk}\big)a_{i'}$ do not in general cancel out. We therefore narrow our analysis to the case in which $M$ is symmetric, that is, $M_{jk}=M_{kj}$ for any $j\neq k$. This means that the advantage spillover effect between any two positions is reciprocal, which is also a common assumption in network analysis. Under this restriction, Proposition~\ref{prop:stability-diagonal} characterizes pairwise stable assignments.

\begin{proposition}[Stability under symmetric $M$]\label{prop:stability-diagonal}
Maintain Assumptions~\ref{ass:network-stability}--\ref{ass:interior-candidate-ordering}. Assume $M$ is symmetric. If $\omega=1$, every assignment is pairwise stable. If $\omega<1$, an assignment $m$ is pairwise stable if and only if for every pair $i\neq i'$
with $j:=m(i)$ and $k:=m(i')$,
\begin{equation}\label{eq:stable-diagonal-ineq}
(a_i-a_{i'})(M_{jj}-M_{kk})\ge 0.
\end{equation}
Equivalently, when $\omega<1$, after ordering abilities so that $a_{i(1)}\le\cdots\le a_{i(n)}$, pairwise stability requires
that the assigned diagonal scores are weakly increasing in the same order:
\[
M_{m(i(1))\,m(i(1))}\le \cdots \le M_{m(i(n))\,m(i(n))},
\]
up to arbitrary permutations within tied blocks.
\end{proposition}
\proofinapp{app:proof-stability-diagonal}

\begin{remark}[Constant diagonal]\label{rem:constant-diagonal-stability}
If $M_{jj}$ is constant across $j$, then \eqref{eq:stable-diagonal-ineq} holds for every assignment,
so every assignment is stable.
\end{remark}

The intuition is simple. Under transferable utility, the off-diagonal spillover terms wash out in any two-agent swap, so the gain from swapping depends only on the two agents’ types and the diagonal terms of the two positions in $M$. This highlights the tension between the planner and the agents, as agents care only about the advantage generated by their own positions, and at most that of the other agent involved in the swap, whereas the planner takes a global view and seeks either to maximize or to minimize total effective advantage.

\paragraph{Empirical implication of pairwise stability.} Proposition~\ref{prop:stability-diagonal} implies that when agents can rearrange assignments through mutually beneficial swaps, observed assignments may reflect private stability rather than the planner's preferred allocation. Under symmetric $M$ and $\omega<1$, stable assignments sort ability by the diagonal entries $M_{jj}$, whereas the planner's unconstrained assignment sorts by the aggregate spillover index $b_j$. Therefore, in settings with informal task trading, reciprocal help, or decentralized assignment adjustment, the observed correlation between ability and $M_{m(i)m(i)}$ should be stronger than the correlation between ability and $b_{m(i)}$. The gap between these two correlations is a direct empirical measure of the tension between private assignment incentives and the planner's global spillover objective.

The stability result is related to evidence that decentralized managerial incentives can distort the allocation of talent inside firms. \citet{BandieraBarankayRasul2009} show that social connections between managers and workers affect productivity and allocation, and \citet{haegele2022talent} provides direct evidence that managers hoard talented workers when their own incentives favor retaining talent within their teams. The model's sharper prediction that stable assignments sort by private positional advantage rather than aggregate spillover value is, however, a novel implication that would require data on both private and aggregate spillover measures.

\vspace{15pt}

There are cases in which the planner need not worry about pairwise stability, as Remark~\ref{rem:constant-diagonal-stability} underscores. More interesting, however, is how the requirement of pairwise stability constrains the planner's maximization problem, and how this depends on spillover and competition structure. We turn to this question in the next subsection.

\subsection{Network structures and the stability loss}\label{subsec:delta-monotone-Q}

The previous subsection shows that, under symmetric $M$ and $\omega<1$, pairwise stability forces the assignment to sort ability by the diagonal entries of $M$. This restriction can prevent the planner from attaining the unconstrained extreme values $B^+$ or $B^-$, and hence reduce achieved output. We refer to this reduction as \emph{stability loss}. This subsection first explains the source of this stability loss and then studies how it depends on the competition structure.

In what follows, we focus on the nondegenerate case in which the relevant
evaluation rules are strictly below the boundary \(\omega=1\). When
\(\omega=1\), the assignment component in the pairwise-stability condition
vanishes, so the boundary case removes the assignment-stability tradeoff that
is central to the analysis. We therefore restrict attention to the case in
which pairwise stability imposes a meaningful constraint on assignment and the
stability-loss problem is nontrivial.

Formally, throughout this subsection assume
\[
s<1, \qquad \widehat\omega(B^+)<1, \qquad \widehat\omega(B^{S,+})<1.
\]
Since \(\widehat\omega(B)\) is increasing in \(B\), these conditions also imply
\[
\widehat\omega(B^-)\le \widehat\omega(B^+)<1
\qquad\text{and}\qquad
\widehat\omega(B^{S,-})\le \widehat\omega(B^{S,+})<1.
\]
Thus the relevant maxima over \([\omega^I,1)\) are attained strictly below
\(\omega=1\), rather than only approached at the boundary.

Recall that $Q$ is the aggregate Katz centrality induced by the competition network $W$, and captures the aggregate strength of competition. For ease of exposition, in this subsection we write the planner's value as a function of $Q$. Let
\[
V^{FB}(Q):=\max_{\omega\in[\omega^I,1)}\ \max_m \mathcal W(\omega,m;Q),
\qquad
V^{SB}(Q):=\max_{\omega\in[\omega^I,1)}\ \max_{m\in\mathcal M^S} \mathcal W(\omega,m;Q),
\]
and define the stability loss by
\[
\Delta(Q):=V^{FB}(Q)-V^{SB}(Q)\ge 0.
\]
Here $\mathcal M^S$ denotes the set of assignments satisfying the stability condition in Proposition~\ref{prop:stability-diagonal}. Also define
\[
B^{S,+}:=\max_{m\in\mathcal M^S} b^\top a^{(m)},
\qquad
B^{S,-}:=\min_{m\in\mathcal M^S} b^\top a^{(m)}.
\]

The lemma below separates the stability loss into two parts. One arises from the loss of total effective advantage, while the other arises from the planner’s inability to use assignment more fully as a device for inducing effort.

\begin{lemma}[Stable and unconstrained branch values]\label{lem:stability-loss}
Maintain Assumptions~\ref{ass:network-stability}--\ref{ass:interior-candidate-ordering}. Assume that $M$ is symmetric. Define the stable cutoff
\[
s^{S,\star}
:=
\frac{B^{S,+}-B^{S,-}}
{\widetilde T(B^{S,-})-\widetilde T(B^{S,+})},
\]
with the convention that $s^{S,\star}:=\widehat\omega(B^{S,+})$ if $\widetilde T(B^{S,-})-\widetilde T(B^{S,+})\le0$.

Then the stability loss is given by
\[
\Delta(Q)
=
\begin{cases}
\big(B^+-B^{S,+}\big)
+s\big[\widetilde T(B^+)-\widetilde T(B^{S,+})\big],
& \text{if } s<s^\star \text{ and } s<s^{S,\star},\\[0.9em]
\big(B^+-B^{S,-}\big)
+s\big[\widetilde T(B^+)-\widetilde T(B^{S,-})\big],
& \text{if } s<s^\star \text{ and } s>s^{S,\star},\\[0.9em]
\big(B^--B^{S,+}\big)
+s\big[\widetilde T(B^-)-\widetilde T(B^{S,+})\big],
& \text{if } s>s^\star \text{ and } s<s^{S,\star},\\[0.9em]
\big(B^--B^{S,-}\big)
+s\big[\widetilde T(B^-)-\widetilde T(B^{S,-})\big],
& \text{if } s>s^\star \text{ and } s>s^{S,\star}.
\end{cases}
\]
At equality, the corresponding branches yield the same value.
\end{lemma}
\proofinapp{app:proof-stability-loss}

Assignment serves two roles. It raises total effective advantage and it induces effort. Pairwise stability restricts both roles by replacing the unconstrained range $[B^-,B^+]$ with the stable range $[B^{S,-},B^{S,+}]$. Lemma~\ref{lem:stability-loss} shows that in all cases, stability loss contains both an effective-advantage component and an effort-inducing component. When the unconstrained and stable problems select different branches, an additional source of loss arises because stability also changes which assignment the planner effectively uses.

Since the second component of stability loss depends on how effective low advantage is at inducing effort, it naturally depends on the competition structure. For example, more intense competition may increase the effectiveness of using low advantage, so when pairwise stability distorts the optimal negative assortative assignment, the resulting loss may be larger. Lemma~\ref{lem:Q-cutoff} and Proposition~\ref{prop:delta-monotone-Q} below make this clear.

\begin{lemma}[Cutoff in competition intensity]\label{lem:Q-cutoff}
Maintain Assumptions~\ref{ass:network-stability}--\ref{ass:interior-candidate-ordering}, and suppose \(\underline\omega e^{1/2}<s\). Define
\[
\Phi(Q)
:=
\frac{1}{B^+-B^-}
\int_{B^-}^{B^+}
\frac{1}{\widehat\omega(B;Q)}\,dB .
\]
Then there exists a unique \(Q^\star>0\) such that
\[
\Phi(Q^\star)=\frac{1}{s}.
\]
Moreover,
\[
Q<Q^\star
\quad\Longleftrightarrow\quad
s<s^\star(Q),
\]
and
\[
Q>Q^\star
\quad\Longleftrightarrow\quad
s>s^\star(Q).
\]
\end{lemma}

\proofinapp{app:proof-lem-Q-cutoff}

\begin{proposition}[Monotonicity of $\Delta(Q)$ in $Q$]
\label{prop:delta-monotone-Q} Maintain Assumptions~\ref%
{ass:network-stability}--\ref{ass:interior-candidate-ordering}. Assume that $%
M$ is symmetric and \(\underline\omega e^{1/2}<\omega^I\le s\). Let $Q^\star$ be defined as in Lemma~\ref{lem:Q-cutoff}. Then, on any
interval of $Q$-values over which the maintained assumptions continue to
hold:

\begin{enumerate}
\item if $Q<Q^\star$, then $\Delta(Q)$ is weakly decreasing in $Q$;

\item if $Q>Q^{\star }$, then $\Delta (Q)$ is weakly increasing in $Q$.
\end{enumerate}
\end{proposition}
\proofinapp{app:proof-delta-monotone-Q}
\paragraph{Empirical implication for stability loss.} The monotonicity result for $\Delta(Q)$ implies that the cost of decentralized or stable assignment constraints depends on the competition environment. When the first-best policy is positive assortative, stronger competition should reduce the output loss from pairwise stability, because evaluation can partly substitute for the planner's preferred assignment. When the first-best policy is negative assortative, stronger competition should increase the output loss from pairwise stability, because the planner loses a more valuable effort-inducing assignment instrument. This prediction can be tested by comparing units with different benchmarking intensity, or by studying reforms that make performance comparison more or less salient while holding the feasible assignment technology fixed.

The comparative statics of the cutoff and of the stability loss are less directly covered by existing empirical work. The prediction that denser benchmarking networks should make negative assortative assignment more attractive, and the prediction that the stability loss changes with competition intensity, appear to be novel implications of the model and suggest a natural agenda for future empirical work.

When competition is stronger, the evaluation rule has a larger systemwide effect on effort. If the first best uses assignment to maximize advantage spillovers, for example because the manager values placing the most talented salespeople in the most connected segments, then stronger competition allows the planner to offset assignment restrictions through evaluation, so the stability loss shrinks. By contrast, if the first best uses assignment as a device for inducing effort, then stronger competition magnifies the gains from that channel. Preventing the planner from using it through pairwise stability therefore becomes more costly, and the stability loss increases. The example below illustrates the mechanism more clearly.

\begin{example}[Low- and high-competition and stability loss]
Set
\[
(a_1,a_2,a_3)=(0,1,2),\qquad
\alpha=0,\qquad
\sigma=3,\qquad
U=10,\qquad
v=0.14,\qquad
\beta=\frac15,
\]
and let
\[
M=
\begin{pmatrix}
2&0&0\\
0&1&2\\
0&2&\frac12
\end{pmatrix}.
\]
Then
\[
\underline\omega=\frac{v\sigma\sqrt{2\pi}}{U}\approx 0.105,
\]
and
\[
b=M^\top\mathbf 1=\left(2,\,3,\,\frac52\right),
\qquad
diag(M)=\left(2,\,1,\,\frac12\right).
\]
Hence pairwise stability places the types $(0,1,2)$ at positions $(3,2,1)$. Thus
\[
B^{S,+}=B^{S,-}=7.
\]
The unconstrained extrema, which correspond to negative and positive assortative assignments with respect to $b$, are
\[
B^+=\frac{17}{2},
\qquad
B^-=\frac{13}{2}.
\]

Now consider
\[
W^L=
\begin{pmatrix}
0&1&0\\
1&0&1\\
0&1&0
\end{pmatrix},
\qquad
W^H=
\begin{pmatrix}
0&1&1\\
1&0&1\\
1&1&0
\end{pmatrix}.
\]
With \(\beta=\frac15\), these imply
\[
Q_L=\mathbf 1^\top(I-\beta W^L)^{-1}\mathbf 1=\frac{95}{23}\approx 4.130,
\qquad
Q_H=\mathbf 1^\top(I-\beta W^H)^{-1}\mathbf 1=5.
\]
Thus \(W^H\) represents stronger competition than \(W^L\), since $W^H$ introduces an additional competitive link between positions 1 and 3. For both competition structures below, one can verify that Assumptions~\ref{ass:network-stability}--\ref{ass:interior-candidate-ordering} hold with $\omega^I=0.18$. The corresponding cutoffs are
\[
s_L^\star\approx 0.261,
\qquad
s_H^\star\approx 0.239.
\]

\paragraph{Case 1: \(s<s^\star\).}
Set \(s=0.20\), so \(s<s^\star\) under both competition structures. By Lemma~\ref{lem:stability-loss},
\[
\Delta(Q)
=
\underbrace{\left(B^+-B^{S,+}\right)}_{\text{blocked advantage allocation}}
-
\underbrace{s\left(\widetilde T(B^{S,+})-\widetilde T(B^+)\right)}_{\text{effort offset}}.
\]
We have
\[
B^+-B^{S,+}=\frac{17}{2}-7=\frac32=1.500,
\]
and
\[
\widetilde T(B^{S,+})-\widetilde T(B^+)
=
\begin{cases}
37.189-31.542=5.647, & \text{under } W^L,\\[0.3em]
51.105-44.914=6.191, & \text{under } W^H.
\end{cases}
\]
The corresponding total effort levels are
\[
\sum_j E_j^\ast=
\begin{cases}
40.042 \text{ at } B^+,\quad 44.189 \text{ at } B^{S,+}, & \text{under } W^L,\\[0.3em]
53.414 \text{ at } B^+,\quad 58.105 \text{ at } B^{S,+}, & \text{under } W^H.
\end{cases}
\]
The resulting stability losses are
\[
\Delta(Q_L)=1.500-0.20\times 5.647=0.371,
\qquad
\Delta(Q_H)=1.500-0.20\times 6.191=0.262.
\]
It can be seen that the high-competition case recovers a larger portion of the stability loss, because it induces more effort under the constrained assignment, relative to the first best.

\paragraph{Case 2: \(s>s^\star\).}
Set \(s=0.45\), so \(s>s^\star\) under both competition structures. By Lemma~\ref{lem:stability-loss},
\[
\Delta(Q)
=
\underbrace{s\left(\widetilde T(B^-)-\widetilde T(B^{S,-})\right)}_{\text{blocked effort inducement}}
-
\underbrace{\left(B^{S,-}-B^-\right)}_{\text{incidental advantage gain}}.
\]
We have
\[
B^{S,-}-B^- = 7-\frac{13}{2}=\frac12=0.500,
\]
and
\[
\widetilde T(B^-)-\widetilde T(B^{S,-})
=
\begin{cases}
39.206-37.189=2.017, & \text{under } W^L,\\[0.3em]
53.280-51.105=2.174, & \text{under } W^H.
\end{cases}
\]
The corresponding total effort levels are
\[
\sum_j E_j^\ast=
\begin{cases}
45.706 \text{ at } B^-,\quad 44.189 \text{ at } B^{S,-}, & \text{under } W^L,\\[0.3em]
59.780 \text{ at } B^-,\quad 58.105 \text{ at } B^{S,-}, & \text{under } W^H.
\end{cases}
\]
The resulting stability losses are
\[
\Delta(Q_L)=0.45\times 2.017-0.500=0.408,
\qquad
\Delta(Q_H)=0.45\times 2.174-0.500=0.478.
\]
It can be seen that the high-competition case incurs a larger stability loss, because stronger competition magnifies the effort-inducement gain under the first best, relative to the constrained case.
\end{example}

\section{Conclusion}
\label{sec:conclusion}
In this paper, we develop a theory of relative performance evaluation in an environment where a planner seeks to maximize output by jointly choosing the evaluation criterion and the assignment of agents to positions. Agents are linked by two distinct structures: a competition structure that determines how performance is compared, and a spillover structure that determines how fixed types propagate across positions. We show that under reasonable assumptions, equilibrium effort is increasing in an agent's Katz centrality in the competition structure and decreasing in effective advantage. Building on this, we show that it is generally not optimal for the planner to evaluate agents solely on the basis of output. Instead, the planner may either increase or decrease the weight placed on effort in evaluation, together with choosing an assignment that either maximizes or minimizes total effective advantage. We then extend the analysis to assignments that are stable to agents' position swaps and show how the resulting stability loss depends on the intensity of competition. Our paper suggests that relative performance evaluation and position assignment should be designed jointly rather than treated as separate organizational choices.

Our analysis also has several limitations that point to useful extensions. First, we treat type as fixed. In practice, agents may learn, accumulate knowledge, or improve their skills over time, and such growth may itself depend on assignment, for example if more central positions generate stronger learning opportunities. Second, treating non-discretionary spillovers as exogenous is a shortcut. In practice, such spillovers may themselves depend on strategic behavior. For example, agents may choose to interact less with direct competitors, which would make the spillover structure endogenous and could alter the planner's problem in non-trivial ways. While extensions in these directions are interesting, they are beyond the scope of this paper, and we leave them for future research.
\bibliographystyle{apalike}
\bibliography{bib1}

%===============================================================
\appendix

\counterwithin{lemma}{section}
\counterwithin{proposition}{section}
\renewcommand{\thelemma}{\thesection.\arabic{lemma}}
\renewcommand{\theproposition}{\thesection.\arabic{proposition}}
\section{Proofs}
%===============================================================

\subsection{Proof of Lemma~\ref{lem:collapse-threshold}}\label{app:proof-lemma1}
\begin{proof}
Fix $(\omega,m)$ and consider the agent at position $j$. The marginal payoff from
increasing $E_j^{(m)}$ is
\[
\frac{\partial}{\partial E_j^{(m)}}
\left(U\Phi(z_j^{(m)})-vE_j^{(m)}\right)
=
U\phi(z_j^{(m)})\frac{\omega}{\sigma}-v
\le
U\phi(0)\frac{\omega}{\sigma}-v.
\]
If $\omega\le\underline\omega$, then
\[
U\phi(0)\frac{\omega}{\sigma}-v\le 0.
\]
Hence each agent's payoff is weakly decreasing in own effort, and strictly
decreasing except possibly at isolated points. Therefore each agent's unique best
response is $E_j^{(m)}=0$, regardless of others' efforts. Thus
$E^{*(m)}(\omega)\equiv 0$ is the unique Stage-2 pure-strategy Nash equilibrium.
\end{proof}

\subsection{Proof of Proposition~\ref{prop:stage2-omega}}\label{app:proof-stage2-omega}

\begin{proof}
Given others' efforts, the effort at position $j$ affects only $x_j^{(m)}$,
hence only $z_j^{(m)}$, and
\[
\frac{\partial z_j^{(m)}}{\partial E_j^{(m)}}=\frac{\omega}{\sigma}>0.
\]
The payoff at position $j$ is
\[
U\Phi(z_j^{(m)})-vE_j^{(m)}.
\]
At any positive best response, the first-order condition is
\[
U\phi(z_j^{(m)})\frac{\omega}{\sigma}=v,
\]
or equivalently,
\[
\phi(z_j^{(m)})=\frac{v\sigma}{U\omega}.
\]
Since $\omega\in[\omega^I,1]$ and $\omega^I>\underline\omega$, we have
\[
\frac{v\sigma}{U\omega}<\phi(0),
\]
so this equation has two roots. We use the positive root, denoted by
$\bar z(\omega)>0$.

The other root is $-\bar z(\omega)$. It cannot be a payoff maximum. Indeed,
the derivative of the payoff with respect to effort has the same sign as 
\begin{equation*}
\phi(z_j)-\phi(\bar z(\omega)).
\end{equation*}
For the normal density, this expression is negative for $z_j<-\bar z(\omega)$%
, positive for $-\bar z(\omega)<z_j<\bar z(\omega)$, and negative for $%
z_j>\bar z(\omega)$. Hence $-\bar z(\omega)$ is a local minimum, whereas $%
\bar z(\omega)$ is the relevant interior maximizer. Therefore, in any
all-active equilibrium, every active agent must satisfy $z_j=\bar z(\omega)$.

Consider the candidate profile satisfying
\[
z^{(m)}=\bar z(\omega)\mathbf 1.
\]
By the definition of $z$,
\[
(I-\beta W)x^{*(m)}(\omega)=\sigma\bar z(\omega)\mathbf 1.
\]
Assumption~\ref{ass:network-stability} implies that $I-\beta W$ is invertible, hence
\[
x^{*(m)}(\omega)
=
\sigma\bar z(\omega)(I-\beta W)^{-1}\mathbf 1
=
\sigma\bar z(\omega)q.
\]
Using
\[
x_j^{*(m)}(\omega)
=
\omega E_j^{*(m)}(\omega)
+
(1-\omega)\big(\theta_j^{(m)}+\alpha\big),
\]
we obtain
\[
E_j^{*(m)}(\omega)
=
\frac{\sigma}{\omega}\bar z(\omega)q_j
-
\frac{1-\omega}{\omega}\big(\theta_j^{(m)}+\alpha\big).
\]

It remains to show that this candidate is the unique best response. By
Assumption~\ref{ass:alpha-bound},
\[
0<
\bar z(\omega)q_j
-
\frac{1-\omega}{\sigma}
\big(\theta_j^{(m)}+\alpha\big)
<
2\bar z(\omega).
\]
Equivalently,
\[
0<
\frac{\omega}{\sigma}E_j^{*(m)}(\omega)
<
2\bar z(\omega).
\]
Thus $E_j^{*(m)}(\omega)>0$. So the candidate solution is feasible and all-active.

Now consider a deviation from
$E_j^{*(m)}(\omega)$ to zero effort, holding the other agents' efforts fixed. The
success margin falls from $\bar z(\omega)$ to
\[
\bar z(\omega)-\frac{\omega}{\sigma}E_j^{*(m)}(\omega).
\]
Since
\[
0\le
\frac{\omega}{\sigma}E_j^{*(m)}(\omega)
\le
2\bar z(\omega),
\]
the interval
\[
\left[
\bar z(\omega)-\frac{\omega}{\sigma}E_j^{*(m)}(\omega),
\bar z(\omega)
\right]
\]
is contained in $[-\bar z(\omega),\bar z(\omega)]$. Hence
\[
\phi(t)\ge \phi(\bar z(\omega))
\]
for every $t$ in this interval. Therefore,
\begin{align*}
\Phi(\bar z(\omega))
-
\Phi\left(
\bar z(\omega)-\frac{\omega}{\sigma}E_j^{*(m)}(\omega)
\right)
&=
\int_{\bar z(\omega)-\frac{\omega}{\sigma}E_j^{*(m)}(\omega)}^{\bar z(\omega)}
\phi(t)\,dt \\
&\ge
\int_{\bar z(\omega)-\frac{\omega}{\sigma}E_j^{*(m)}(\omega)}^{\bar z(\omega)}
\phi(\bar z(\omega))\,dt \\
&=
\phi(\bar z(\omega))
\left[
\bar z(\omega)
-
\left(
\bar z(\omega)-\frac{\omega}{\sigma}E_j^{*(m)}(\omega)
\right)
\right] \\
&=
\phi(\bar z(\omega))
\frac{\omega}{\sigma}E_j^{*(m)}(\omega).
\end{align*}
Using the first-order condition
\[
v=U\phi(\bar z(\omega))\frac{\omega}{\sigma},
\]
we have
\[
\phi(\bar z(\omega))
\frac{\omega}{\sigma}E_j^{*(m)}(\omega)
=
\frac{v}{U}E_j^{*(m)}(\omega).
\]
Therefore,
\[
\Phi(\bar z(\omega))
-
\Phi\left(
\bar z(\omega)-\frac{\omega}{\sigma}E_j^{*(m)}(\omega)
\right)
\ge
\frac{v}{U}E_j^{*(m)}(\omega).
\]
Multiplying both sides by $U$ gives
\[
U\Phi(\bar z(\omega))
-
U\Phi\left(
\bar z(\omega)-\frac{\omega}{\sigma}E_j^{*(m)}(\omega)
\right)
\ge
vE_j^{*(m)}(\omega).
\]
Rearranging yields
\[
U\Phi(\bar z(\omega))-vE_j^{*(m)}(\omega)
\ge
U\Phi\left(
\bar z(\omega)-\frac{\omega}{\sigma}E_j^{*(m)}(\omega)
\right).
\]
Thus deviating to zero effort is not profitable.

Finally, we show that no one deviates to any other positive effort levels. For any
$E_j>0$, the derivative of the payoff with respect to $E_j$ is
\[
U\phi(z_j)\frac{\omega}{\sigma}-v.
\]
Using
\[
v=U\phi(\bar z(\omega))\frac{\omega}{\sigma},
\]
as discussed earlier, this derivative has the same sign as $\phi(z_j)-\phi(\bar z(\omega))$. Moreover, Assumption~\ref{ass:alpha-bound} implies that, at zero effort,
\[
z_j(0)
=
\bar z(\omega)-\frac{\omega}{\sigma}E_j^{*(m)}(\omega)
\ge
-\bar z(\omega).
\]
Since \(E_j^{*(m)}(\omega)>0\), we also have
\[
z_j(0)<\bar z(\omega).
\]
As \(E_j\) increases from \(0\) to \(E_j^{*(m)}(\omega)\), \(z_j(E_j)\)
therefore moves from a value in \([-\bar z(\omega),\bar z(\omega))\) to
\(\bar z(\omega)\). Since the derivative of the payoff has the same sign as
\(\phi(z_j)-\phi(\bar z(\omega))\), it is strictly positive whenever
\[
-\bar z(\omega)<z_j(E_j)<\bar z(\omega),
\]
and equals zero only at \(z_j(E_j)=\pm \bar z(\omega)\). Hence the payoff is
strictly increasing on \((0,E_j^{*(m)}(\omega))\), up to the possible endpoint
case \(z_j(0)=-\bar z(\omega)\), and is strictly decreasing once
\(z_j>\bar z(\omega)\). Therefore no positive effort level other than
\(E_j^{*(m)}(\omega)\) can be optimal.

Together with the previous comparison showing that zero effort does not give
a higher payoff against the candidate profile, this proves that $E_j^{*(m)}(\omega)$ is a best response for each position $j$. Hence the
candidate profile is a pure-strategy Nash equilibrium.

Finally, suppose there is an all-active pure-strategy equilibrium. Since
every agent chooses positive effort, the preceding first-order condition
must hold for each agent. The negative root is not a payoff maximum, so each
agent must satisfy $z_{j}=\bar{z}(\omega )$. Therefore 
\begin{equation*}
z^{(m)}=\bar{z}(\omega )\mathbf{1},
\end{equation*}%
which uniquely implies 
\begin{equation*}
x^{\ast (m)}(\omega )=\sigma \bar{z}(\omega )(I-\beta W)^{-1}\mathbf{1}%
=\sigma \bar{z}(\omega )q.
\end{equation*}%
The effort vector is then uniquely pinned down by 
\begin{equation*}
x_{j}^{\ast (m)}(\omega )=\omega E_{j}^{\ast (m)}(\omega )+(1-\omega )\bigl(%
\theta _{j}^{(m)}+\alpha \bigr).
\end{equation*}%
Thus the candidate profile is the unique all-active pure-strategy
equilibrium.
\end{proof}

\subsection{Proof of Proposition~\ref{prop:opt-assign-given-omega}}\label{app:proof-opt-assign-given-omega}

\begin{proof}
Fix $\omega\in[\omega^I,1]$ and recall $b:=M^\top\mathbf 1$. By
Lemma~\ref{lem:reduction-b}, conditional on $\omega$ the planner's assignment
problem is
\[
\max_m b^\top a^{(m)} \quad \text{if }\omega>s,
\qquad
\min_m b^\top a^{(m)} \quad \text{if }\omega<s,
\]
and the objective is independent of $m$ if $\omega=s$.

Let $(i(1),\dots,i(n))$ be an ordering such that
\[
a_{i(1)}\le\cdots\le a_{i(n)},
\]
and let $(j(1),\dots,j(n))$ be an ordering such that
\[
b_{j(1)}\le\cdots\le b_{j(n)}.
\]

Consider any assignment $m$ and two agents $i,i'$ assigned to positions
$j=m(i)$ and $k=m(i')$. Let $m^{ii'}$ denote the assignment obtained by swapping
$i$ and $i'$. Then
\[
b^\top a^{(m^{ii'})}-b^\top a^{(m)}
=
b_j a_{i'}+b_k a_i-(b_j a_i+b_k a_{i'})
=
(b_k-b_j)(a_i-a_{i'}).
\]

If $\omega>s$, the planner maximizes $b^\top a^{(m)}$. If an assignment has an
inversion, i.e., there are two agents $i,i'$ such that $a_i>a_{i'}$ but
$b_{m(i)}<b_{m(i')}$, then the swap formula above shows that swapping them
strictly increases $b^\top a^{(m)}$. Hence no maximizing assignment can have such
an inversion, except within tied values of $\{a_i\}$ or $\{b_j\}$. Therefore an
assignment maximizes $b^\top a^{(m)}$ if and only if it matches higher abilities
to weakly higher values of $b_j$, up to arbitrary permutations within tied
blocks. This is exactly the first case in \eqref{eq:opt-assignment-given-omega}.

If $\omega<s$, the planner minimizes $b^\top a^{(m)}$. If an assignment fails to
match higher abilities to weakly lower values of $b_j$, then there are two agents
$i,i'$ such that $a_i<a_{i'}$ but $b_{m(i)}<b_{m(i')}$. The swap formula shows
that swapping them strictly decreases $b^\top a^{(m)}$. Hence no minimizing
assignment can have such an inversion, except within tied values of $\{a_i\}$ or
$\{b_j\}$. Therefore an assignment minimizes $b^\top a^{(m)}$ if and only if it
matches higher abilities to weakly lower values of $b_j$, up to arbitrary
permutations within tied blocks. This is exactly the second case in
\eqref{eq:opt-assignment-given-omega}.

Finally, if $\omega=s$, the coefficient on $b^\top a^{(m)}$ is zero by
Lemma~\ref{lem:reduction-b}, so every assignment is planner-optimal. This is
exactly the third case in \eqref{eq:opt-assignment-given-omega}.

This completes the proof.
\end{proof}

\subsection{Proof of Lemma~\ref{lem:positive-only}}\label{app:proof-lemma3}

\begin{proof}
Since $s<\omega^I$ and the planner's choice is restricted to
$\omega\in[\omega^I,1]$, every feasible $\omega$ satisfies $\omega>s$.
Therefore, by Proposition~\ref{prop:opt-assign-given-omega}, the planner always
chooses a positive assortative assignment conditional on any feasible $\omega$.
Hence the negative-assortative branch is empty.

It remains to choose $\omega$ on the positive-assortative branch. On this branch,
the planner's value is
\[
V^+(\omega)
=
s\,\frac{\sigma}{\omega}\bar z(\omega)\,Q
-\;s\,\frac{1-\omega}{\omega}\,n\alpha
+\Big(1-\frac{s}{\omega}\Big)B^+,
\qquad \omega\in[\omega^I,1].
\]
By definition,
\[
\omega_+^\ast
=
\min\left\{1,\max\left\{\omega^I,\widehat\omega(B^+)\right\}\right\}
\]
is the maximizer of $V^+(\omega)$ over $[\omega^I,1]$. To see this, write
\[
V^+(\omega)
=
B^+
+
sn\alpha
+
\frac{s}{\omega}
\left(
\sigma Q\bar z(\omega)-B^+-n\alpha
\right).
\]
Since
\[
\bar z(\omega)=\sqrt{2\log\left(\frac{\omega}{\underline\omega}\right)},
\]
differentiating w.r.t. $\omega$ gives
\[
\bar z'(\omega)
=
\frac{1}{2}
\left(2\log\left(\frac{\omega}{\underline\omega}\right)\right)^{-1/2}
\cdot \frac{2}{\omega}
=
\frac{1}{\omega\bar z(\omega)}.
\]
we then have
\[
\frac{dV^+(\omega)}{d\omega}
=
\frac{s}{\omega^2}
\left[
-\sigma Q\bar z(\omega)
+
(B^+ + n\alpha)
+
\frac{\sigma Q}{\bar z(\omega)}
\right].
\]
From the first order condition that $\frac{dV^+(\omega)}{d\omega}=0$, we get that the stationary point satisfies
\[
\bar z(\omega)-\frac{1}{\bar z(\omega)}
=
\frac{B^+ + n\alpha}{\sigma Q}
=
K(B^+).
\]
Because the function \(z-1/z\) is strictly increasing from $-\infty$ to $+\infty$ on \(z>0\), this equation
has the unique positive solution \(z(B^+)\). Let \(z=\bar z(\omega)\). Then \(z\) solves
\[
z-\frac{1}{z}=K(B^+),
\]
or equivalently,
\[
z^2-K(B^+)z-1=0.
\]
The unique positive solution is
\[
z(B^+)
=
\frac{K(B^+)+\sqrt{K(B^+)^2+4}}{2}.
\]
Therefore the stationary point must satisfy
\[
\bar z(\omega)=z(B^+).
\]
Since
\[
\bar z(\omega)=\sqrt{2\log\left(\frac{\omega}{\underline\omega}\right)},
\]
we get
\[
z(B^+)^2
=
2\log\left(\frac{\omega}{\underline\omega}\right).
\]
Hence
\[
\omega
=
\underline\omega\exp\left(\frac{z(B^+)^2}{2}\right).
\]
Therefore the unique stationary point is
\[
\widehat\omega(B^+)
=
\underline\omega\exp\left(\frac{z(B^+)^2}{2}\right).
\]
To see that this stationary point is a maximum, note that
\[
\frac{dV^+(\omega)}{d\omega}
=
\frac{s\sigma Q}{\omega^2}
\left[
-\bar z(\omega)
+
K(B^+)
+
\frac{1}{\bar z(\omega)}
\right].
\]
Since $\bar z(\omega)$ is strictly increasing in $\omega$, and the function
\[
h(z):=-z+K(B^+)+\frac{1}{z}
\]
is strictly decreasing in $z>0$, the term
\[
h(\bar z(\omega))
\]
is strictly decreasing in $\omega$. Hence $\frac{dV^+(\omega)}{d\omega}$ changes sign from positive to negative at the unique solution $z(B^+)$. Hence $V^+(\omega)$ is increasing for
$\omega<\widehat\omega(B^+)$ and decreasing for
$\omega>\widehat\omega(B^+)$. Therefore the maximizer of \(V^+(\omega)\) over
the restricted interval \([\omega^I,1]\) is the projection of
\(\widehat\omega(B^+)\) onto this interval:
\[
\omega_+^\ast
=
\min\left\{1,\max\left\{\omega^I,\widehat\omega(B^+)\right\}\right\}.
\] Therefore the policy
\[
(\omega,m)=(\omega_+^\ast,m^+)
\]
is optimal for the planner.
\end{proof}

\subsection{Proof of Proposition~\ref{prop:omega-cutoff}}\label{app:proof-omega-cutoff}

\begin{proof}
Suppose $\omega^I\le s$. Then the feasible negative-assortative branch is
\[
\omega\in[\omega^I,s],
\]
and the feasible positive-assortative branch is
\[
\omega\in[s,1].
\]
By Proposition~\ref{prop:opt-assign-given-omega}, the planner chooses $m^-$ on
the negative branch and $m^+$ on the positive branch.

For any branch value $B\in\{B^-,B^+\}$, the planner's value can be written as
\[
V_B(\omega)
=
s\,\frac{\sigma}{\omega}\bar z(\omega)\,Q
-\;s\,\frac{1-\omega}{\omega}\,n\alpha
+\Big(1-\frac{s}{\omega}\Big)B.
\]
Equivalently,
\[
V_B(\omega)
=
B
+
s\left[
n\alpha+
\frac{\sigma Q\bar z(\omega)-B-n\alpha}{\omega}
\right].
\]
By the same argument as in the proof of Lemma~\ref{lem:positive-only}, replacing
$B^+$ by a generic $B$, the unique stationary point of $V_B(\omega)$ is
\[
\widehat\omega(B)
=
\underline\omega\exp\left(\frac{z(B)^2}{2}\right),
\]
where
\[
K(B):=\frac{B+n\alpha}{\sigma Q},
\qquad
z(B):=\frac{K(B)+\sqrt{K(B)^2+4}}{2}.
\]
Moreover, $V_B(\omega)$ is increasing before $\widehat\omega(B)$ and decreasing
after $\widehat\omega(B)$.

Since $B^+\ge B^-$, we have $K(B^+)\ge K(B^-)$. Since the map
\[
K\mapsto \frac{K+\sqrt{K^2+4}}{2}
\]
is strictly increasing, it follows that
\[
z(B^+)\ge z(B^-),
\]
and therefore
\[
\widehat\omega(B^+)\ge \widehat\omega(B^-).
\]

The maximizer on the positive-assortative branch $[s,1]$ is
\[
\omega_+^\ast
=
\min\left\{1,\max\left\{s,\widehat\omega(B^+)\right\}\right\}.
\]
Since $\omega^I\le s$, this is equivalent to
\[
\omega_+^\ast
=
\min\left\{1,\max\left\{\max\{s,\omega^I\},\widehat\omega(B^+)\right\}\right\}.
\]
Similarly, the maximizer on the negative-assortative branch $[\omega^I,s]$ is
\[
\omega_-^\ast
=
\min\left\{s,\max\left\{\omega^I,\widehat\omega(B^-)\right\}\right\}.
\]

For any scalar \(r\in[\omega^I,1]\), write \(V^-_r\) and \(V^+_r\) for the
negative- and positive-assortative branch values when the true output weight is
\(r\).

First consider the case
\[
1<\widehat\omega(B^-)\le \widehat\omega(B^+).
\]
Then both stationary points lie above the feasible interval. Hence $V^-(\omega)$
is increasing on $[\omega^I,s]$ and $V^+(\omega)$ is increasing on $[s,1]$.
Therefore
\[
\omega_-^\ast=s,
\qquad
\omega_+^\ast=1.
\]
At $\omega=s$, the assignment coefficient is
\[
1-\frac{s}{s}=0,
\]
so the planner is indifferent over assignments. Thus the best value on the
negative branch equals the value at $\omega=s$ under any assignment. Since
$V^+(\omega)$ is increasing on $[s,1]$, we have
\[
V^+(1)\ge V^+(s)=V^-(s).
\]
Therefore the policy $(\omega_+^\ast,m^+)=(1,m^+)$ is optimal for the planner.

Now consider Case 2:
\[
\widehat\omega(B^+)<\omega^I\le s.
\]
Then
\[
\widehat\omega(B^-)\le \widehat\omega(B^+)<\omega^I,
\]
so both stationary points lie below the active-contest region. Since each branch
value is decreasing after its stationary point, the negative branch is maximized
at
\[
\omega_-^\ast=\omega^I,
\]
and the positive branch is maximized at
\[
\omega_+^\ast=s.
\]

At $\omega=s$, the assignment coefficient is
\[
1-\frac{s}{s}=0.
\]
Hence the planner is indifferent over assignments at $\omega=s$, and therefore
\[
V^-_s(s)=V^+_s(s).
\]
Since \(\widehat\omega(B^+)<\omega^I\le s\), the positive-branch value is
decreasing on \([s,1]\), so its maximum is attained at \(\omega=s\). Since
\(\widehat\omega(B^-)<\omega^I\le s\), the negative-branch value is decreasing on
\([\omega^I,s]\), so its maximum is attained at \(\omega=\omega^I\). Therefore
\[
V^-_s(\omega^I)\ge V^-_s(s)=V^+_s(s).
\]
Thus the negative-assortative policy \((\omega_-^\ast,m^-)\) is optimal. If \(s=\omega^I\), then \(\omega=s=\omega^I\) is the switching point, and we have $\omega_-^\ast=s=\omega^I$. So every policy \((\omega^I,m)\) is optimal, regardless of the assignment \(m\). If \(s>\omega^I\), \((\omega_-^\ast,m^-)\) is optimal.

Finally consider Case 3:
\[
\widehat\omega(B^-)\le 1
\qquad\text{and}\qquad
\omega^I\le \widehat\omega(B^+).
\]
Define the fixed reference points
\[
\bar\omega^-:=
\begin{cases}
\widehat\omega(B^-), & \text{if } \omega^I\le \widehat\omega(B^-)\le 1,\\
\omega^I, & \text{if } \widehat\omega(B^-)<\omega^I,
\end{cases}
\]
and
\[
\bar\omega^+:=
\begin{cases}
\widehat\omega(B^+), & \text{if } \omega^I\le \widehat\omega(B^+)\le 1,\\
1, & \text{if } \widehat\omega(B^+)>1.
\end{cases}
\]
By the monotonicity of $V_B(\omega)$ around its stationary point,
\[
\bar\omega^-\le \bar\omega^+.
\]

By construction,
\[
V^-_r(\bar\omega^-)=B^-+r\widetilde T^-,
\qquad
V^+_r(\bar\omega^+)=B^+ + r\widetilde T^+.
\]
Indeed, when the relevant reference point is a stationary point, this follows
from the first-order condition
\[
\sigma Q z(B)-B-n\alpha=\frac{\sigma Q}{z(B)},
\]
Substituting the expression for \(\sigma Q z(B)-B-n\alpha\) obtained
from this condition into
\[
n\alpha+
\frac{\sigma Q z(B)-B-n\alpha}{\widehat\omega(B)}
\]
yields the result. When the relevant reference point is a boundary point, it follows directly from
substitution into $V_B(\omega)$.

We now show that \(\widetilde T^- - \widetilde T^+>0\). Define
\[
G(B,\omega)
:=
n\alpha+
\frac{\sigma Q\bar z(\omega)-B-n\alpha}{\omega}.
\]
By the definitions above,
\[
\widetilde T^-=\max_{\omega\in[\omega^I,1]}G(B^-,\omega),
\qquad
\widetilde T^+=\max_{\omega\in[\omega^I,1]}G(B^+,\omega).
\]
For any fixed \(\omega\in[\omega^I,1]\),
\[
G(B^-,\omega)-G(B^+,\omega)
=
\frac{B^+-B^-}{\omega}.
\]
Since \(B^+>B^-\), this difference is strictly positive for every feasible
\(\omega\). Therefore
\[
\widetilde T^- > \widetilde T^+.
\]

Define
\[
F(r):=V^+_r(\bar\omega^+)-V^-_r(\bar\omega^-)
=
B^+-B^-+r(\widetilde T^+-\widetilde T^-).
\]

At $r=\bar\omega^-$, the point $\omega=\bar\omega^-$ is the switching point, so
the assignment coefficient vanishes:
\[
1-\frac{\bar\omega^-}{\bar\omega^-}=0.
\]
Hence
\[
V^+_{\bar\omega^-}(\bar\omega^-)
=
V^-_{\bar\omega^-}(\bar\omega^-).
\]
Since $\bar\omega^+$ maximizes the positive-side value over $[\bar\omega^-,1]$,
\[
V^+_{\bar\omega^-}(\bar\omega^+)
\ge
V^+_{\bar\omega^-}(\bar\omega^-)
=
V^-_{\bar\omega^-}(\bar\omega^-).
\]
Therefore
\[
F(\bar\omega^-)\ge0.
\]

At $r=\bar\omega^+$, the point $\omega=\bar\omega^+$ is the switching point, so
the assignment coefficient vanishes:
\[
1-\frac{\bar\omega^+}{\bar\omega^+}=0.
\]
Hence
\[
V^+_{\bar\omega^+}(\bar\omega^+)
=
V^-_{\bar\omega^+}(\bar\omega^+).
\]
Since $\bar\omega^-$ maximizes the negative-side value over
$[\omega^I,\bar\omega^+]$,
\[
V^-_{\bar\omega^+}(\bar\omega^-)
\ge
V^-_{\bar\omega^+}(\bar\omega^+)
=
V^+_{\bar\omega^+}(\bar\omega^+).
\]
Therefore
\[
F(\bar\omega^+)\le0.
\]

We now return to the actual output weight $s$.

First, suppose
\[
s<\bar\omega^-.
\]
Then the negative branch is $[\omega^I,s]$ and cannot reach $\bar\omega^-$. Since
$V^-_s(\omega)$ is increasing on $[\omega^I,s]$, the best feasible point on the
negative branch is $\omega=s$. At $\omega=s$, the assignment coefficient is zero:
\[
1-\frac{s}{s}=0.
\]
Thus
\[
V^-_s(s)=V^+_s(s).
\]
Since $\bar\omega^+$ maximizes the positive-side value over $[s,1]$,
\[
V^+_s(\bar\omega^+)\ge V^+_s(s)=V^-_s(s).
\]
Therefore the positive-assortative policy is optimal whenever $s<\bar\omega^-$.

Second, suppose
\[
s>\bar\omega^+.
\]
Then the positive branch is $[s,1]$ and cannot reach $\bar\omega^+$. Since
$V^+_s(\omega)$ is decreasing after $\bar\omega^+$, the best feasible point on
the positive branch is $\omega=s$. At $\omega=s$, the assignment coefficient is
zero:
\[
1-\frac{s}{s}=0.
\]
Thus
\[
V^+_s(s)=V^-_s(s).
\]
Since $\bar\omega^-$ maximizes the negative-side value over $[\omega^I,s]$,
\[
V^-_s(\bar\omega^-)\ge V^-_s(s)=V^+_s(s).
\]

Therefore the negative-assortative policy is optimal whenever $s>\bar\omega^+$.

Third, suppose
\[
\bar\omega^-\le s\le \bar\omega^+.
\]
Then $\bar\omega^-$ is feasible on the negative branch and $\bar\omega^+$ is
feasible on the positive branch. Hence the comparison of the two branch-optimal
values is exactly
\[
F(s)=B^+-B^-+s(\widetilde T^+-\widetilde T^-).
\]

Since \(\widetilde T^- - \widetilde T^+>0\), \(F(r)\) is strictly decreasing in
\(r\). The unique solution to \(F(r)=0\) is
\[
s^\star
=
\frac{B^+-B^-}{\widetilde T^- - \widetilde T^+}.
\]
Since \(F(\bar\omega^-)\ge0\) and \(F(\bar\omega^+)\le0\), we have
\[
\bar\omega^-\le s^\star\le \bar\omega^+.
\]
Additionally, if \(s<\bar\omega^-\), then \(s<s^\star\), and if
\(s>\bar\omega^+\), then \(s>s^\star\).   Therefore, combining the three cases together shows that the positive-assortative policy is optimal if
\(s<s^\star\), the negative-assortative policy is optimal if \(s>s^\star\), and
the planner is indifferent if \(s=s^\star\).

This completes the proof.
\end{proof}

\subsection{Extending the Choice Set Below $\omega^I$}\label{app:below-omegaI}

This section provides auxiliary results, aimed at simplifying the Assumption~\ref{ass:interior-candidate-ordering} and giving a discussion regarding the restriction to $\omega\in[\omega^I,1]$. We first show that Assumption~\ref{ass:interior-candidate-ordering} implies that
the relevant unconstrained candidate evaluation weights lie inside the admissible
region considered in Proposition~\ref{prop:omega-cutoff}.

\begin{lemma}\label{lem:assumption-implies-interior-ordering}
Under Assumption~\ref{ass:interior-candidate-ordering},
\[
\omega^I
\le
\widehat\omega(B^-)
\le
\widehat\omega(B^+)
\le
1.
\]
\end{lemma}

\begin{proof}
Recall that $\widehat\omega(B)$ is defined by
\[
\widehat\omega(B)
=
\underline\omega\exp\left(\frac{z(B)^2}{2}\right),
\]
where $z(B)>0$ solves
\[
z(B)-\frac{1}{z(B)}
=
\frac{B+n\alpha}{\sigma Q}.
\]
The function
\[
z\mapsto z-\frac{1}{z}
\]
is strictly increasing on $z>0$, because its derivative is
\[
1+\frac{1}{z^2}>0.
\]
The function
\[
z\mapsto \underline\omega\exp\left(\frac{z^2}{2}\right)
\]
is also strictly increasing on $z>0$. Hence $\widehat\omega(B)$ is increasing in
$B$.

Since $B^+\ge B^-$, it follows immediately that
\[
\widehat\omega(B^-)\le \widehat\omega(B^+).
\]

It remains to show the two boundary inequalities. By
Assumption~\ref{ass:interior-candidate-ordering},
\[
\sigma Q\left[
\bar z(\omega^I)-\frac{1}{\bar z(\omega^I)}
\right]
\le
B^-+n\alpha.
\]
Dividing by $\sigma Q>0$ gives
\[
\bar z(\omega^I)-\frac{1}{\bar z(\omega^I)}
\le
\frac{B^-+n\alpha}{\sigma Q}
=
z(B^-)-\frac{1}{z(B^-)}.
\]
Since $z-1/z$ is strictly increasing on $z>0$, this implies
\[
\bar z(\omega^I)\le z(B^-).
\]
Therefore,
\[
\omega^I
=
\underline\omega
\exp\left(\frac{\bar z(\omega^I)^2}{2}\right)
\le
\underline\omega
\exp\left(\frac{z(B^-)^2}{2}\right)
=
\widehat\omega(B^-).
\]

Similarly, Assumption~\ref{ass:interior-candidate-ordering} gives
\[
B^+ + n\alpha
\le
\sigma Q\left[
\bar z(1)-\frac{1}{\bar z(1)}
\right].
\]
Dividing by $\sigma Q>0$ gives
\[
z(B^+)-\frac{1}{z(B^+)}
=
\frac{B^+ + n\alpha}{\sigma Q}
\le
\bar z(1)-\frac{1}{\bar z(1)}.
\]
Again, since $z-1/z$ is strictly increasing on $z>0$, we obtain
\[
z(B^+)\le \bar z(1).
\]
Thus,
\[
\widehat\omega(B^+)
=
\underline\omega
\exp\left(\frac{z(B^+)^2}{2}\right)
\le
\underline\omega
\exp\left(\frac{\bar z(1)^2}{2}\right)
=
1.
\]
Combining these inequalities yields
\[
\omega^I
\le
\widehat\omega(B^-)
\le
\widehat\omega(B^+)
\le
1.
\]
\end{proof}

 In each of the two propositions below, we state a sufficient condition under which restricting the planner's choice of $\omega$ to $[\omega^I,1]$ is without
loss. That is, even if an equilibrium, whether interior or corner, exists
for some $\omega<\omega^I$, the planner does not strictly prefer it to the best
policy already attainable in the active-contest region.

The first sufficient condition controls the behavior of every possible
active set below $\omega^I$. Recall that $J$ denotes the set of positions. For
any assignment $m$, write $c^{(m)}:=\theta^{(m)}+\alpha\mathbf 1$. For any nonempty subset $A\subseteq J$, define
\[
Q_A
:=
\mathbf 1_A^\top (I-\beta W_{AA})^{-1}\mathbf 1_A,
\]
and
\[
H_A^{(m)}
:=
\mathbf 1_A^\top c_A^{(m)}
-
\mathbf 1_A^\top (I-\beta W_{AA})^{-1}\beta W_{A,-A}c_{-A}^{(m)}.
\]
where $\mathbf 1_A$ denotes a vector of ones indexed by positions in $A$,
$c_A^{(m)}$ and $c_{-A}^{(m)}$ denote the subvectors of $c^{(m)}$ corresponding
to positions in $A$ and $J\setminus A$, respectively, $W_{AA}$ is the submatrix
of $W$ with rows and columns in $A$, and $W_{A,-A}$ is the submatrix of $W$ with
rows in $A$ and columns in $J\setminus A$.\footnote{Since $W\ge0$ and $W_{AA}$ is a principal submatrix of $W$,
Perron-Frobenius monotonicity implies $\rho(W_{AA})\le\rho(W)$; hence
Assumption~\ref{ass:network-stability} implies that $I-\beta W_{AA}$ is
invertible.} Then define
\[
\Xi
:=
\min_{m\in\mathcal M}
\min_{\varnothing\neq A\subseteq J}
\frac{H_A^{(m)}}{\sigma Q_A},
\qquad
\tau_I
:=
\bar z(\omega^I)-\frac{1}{\bar z(\omega^I)}.
\]

\begin{proposition}\label{prop:below-omegaI-Xi}
Maintain Assumptions~\ref{ass:network-stability}, \ref{ass:alpha-bound}, and
\ref{ass:interior-candidate-ordering}. Suppose $\omega^I\le s$. If
\[
\Xi\ge \tau_I,
\]
then allowing the planner to choose $\omega<\omega^I$ does not change the optimal
policy characterized in Proposition~\ref{prop:omega-cutoff}.
\end{proposition}

\begin{proof}
By Lemma~\ref{lem:assumption-implies-interior-ordering},
\[
\omega^I
\le
\widehat\omega(B^-)
\le
\widehat\omega(B^+)
\le
1.
\]

First consider $\omega\le \underline\omega$. By
Lemma~\ref{lem:collapse-threshold}, the unique Stage-2 equilibrium is zero
effort. Hence the planner's value is no greater than $B^+$, which is weakly
dominated by the value of the positive-assortative candidate policy in the
active-contest region. Therefore no $\omega\le\underline\omega$ can strictly
improve on the policy characterized in Proposition~\ref{prop:omega-cutoff}.

Now consider $\omega\in(\underline\omega,\omega^I)$. Fix an assignment $m$ and a
Stage-2 equilibrium, if one exists. Let $A$ be its active set. Active agents
satisfy $z_j=\bar z(\omega)$, while inactive agents exert zero effort. Hence the
aggregate effort in this equilibrium is given by \[
S_A^{(m)}(\omega)
:=
\sum_{j\in A}E_j(\omega)
=
\frac{\sigma Q_A}{\omega}\bar z(\omega)
-
\frac{1-\omega}{\omega}H_A^{(m)}.
\]

Differentiating $S_A^{(m)}(\omega)$ gives
\[
\frac{dS_A^{(m)}(\omega)}{d\omega}
=
\frac{\sigma Q_A}{\omega^2}
\left[
\frac{1}{\bar z(\omega)}
-
\bar z(\omega)
+
\frac{H_A^{(m)}}{\sigma Q_A}
\right].
\]
Since $\Xi\ge\tau_I$, we have
\[
\frac{H_A^{(m)}}{\sigma Q_A}
\ge
\bar z(\omega^I)-\frac{1}{\bar z(\omega^I)}
\]
for every $m$ and every nonempty $A$. Since
\[
\bar z(\omega)-\frac{1}{\bar z(\omega)}
\]
is increasing in $\omega$, for every $\omega<\omega^I$ we have
\[
\bar z(\omega)-\frac{1}{\bar z(\omega)}
<
\bar z(\omega^I)-\frac{1}{\bar z(\omega^I)}
\le
\frac{H_A^{(m)}}{\sigma Q_A}.
\]
Therefore
\[
\frac{dS_A^{(m)}(\omega)}{d\omega}>0
\]
on $(\underline\omega,\omega^I)$. Hence
\[
S_A^{(m)}(\omega)
\le
S_A^{(m)}(\omega^I).
\]

At \(\omega^I\), the full-active profile is feasible by
Assumption~\ref{ass:alpha-bound}. We now show that its aggregate effort weakly
dominates the aggregate effort associated with any active set \(A\). Let
\(E_A^A(\omega^I)\) denote the effort vector of the agents in \(A\) under the
active-set profile, and let \(E^J(\omega^I)=E^{*(m)}(\omega^I)\) denote the
full-active effort vector. By the active-set formula,
\[
E_A^A(\omega^I)
=
\frac{1}{\omega^I}
\left[
\sigma\bar z(\omega^I)(I-\beta W_{AA})^{-1}\mathbf 1_A
-
(1-\omega^I)c_A^{(m)}
+
(1-\omega^I)(I-\beta W_{AA})^{-1}\beta W_{A,-A}c_{-A}^{(m)}
\right].
\]
By contrast, the full-active effort vector satisfies
\[
E_A^J(\omega^I)
=
\frac{1}{\omega^I}
\left[
\sigma\bar z(\omega^I)q_A
-
(1-\omega^I)c_A^{(m)}
\right].
\]
Using
\[
q_A
=
(I-\beta W_{AA})^{-1}
\left[
\mathbf 1_A+\beta W_{A,-A}q_{-A}
\right],
\]
we obtain
\[
E_A^J(\omega^I)-E_A^A(\omega^I)
=
(I-\beta W_{AA})^{-1}\beta W_{A,-A}E_{-A}^J(\omega^I).
\]
Since \((I-\beta W_{AA})^{-1}\ge0\), \(W_{A,-A}\ge0\), and
\(E_{-A}^J(\omega^I)\ge0\), it follows that
\[
E_A^J(\omega^I)\ge E_A^A(\omega^I).
\]
Therefore
\[
S_A^{(m)}(\omega^I)
=
\mathbf 1_A^\top E_A^A(\omega^I)
\le
\mathbf 1_A^\top E_A^J(\omega^I)
\le
\sum_{j=1}^n E_j^{*(m)}(\omega^I).
\]
Therefore any equilibrium with $\omega<\omega^I$ satisfies
\[
\mathcal W(\omega,m)
\le
s\sum_j E_j^{*(m)}(\omega^I)
+
(1-s)b^\top a^{(m)}
=
\mathcal W(\omega^I,m).
\]

Since $\omega^I\le s$, the assignment coefficient at $\omega^I$ is nonpositive.
Thus, conditional on $\omega^I$, the planner weakly prefers a negative
assortative assignment. Hence
\[
\mathcal W(\omega^I,m)
\le
\mathcal W(\omega^I,m^-).
\]
Finally, $(\omega^I,m^-)$ is feasible in the restricted problem
$\omega\in[\omega^I,1]$, and Proposition~\ref{prop:omega-cutoff} characterizes
the optimum over that restricted region. Thus
\[
\mathcal W(\omega^I,m^-)
\le
\max\left\{
\mathcal{W}(\omega_+^\ast,m^+),
\mathcal{W}(\omega_-^\ast,m^-)
\right\}.
\]
Therefore no equilibrium induced by any $\omega<\omega^I$ gives the planner a
strictly higher value than the best policy characterized in
Proposition~\ref{prop:omega-cutoff}. Allowing $\omega<\omega^I$ does not change
the planner's optimal policy.
\end{proof}

Intuitively, the only reason a corner equilibrium could be attractive to
the planner is that a small group of active agents may be pushed to exert very
high effort. This can happen when active agents have low effective advantage but
compete against inactive agents with high effective advantage. The
condition $\Xi\ge\tau_I$ rules out this force by requiring every possible active
set to have enough own effective advantage, relative to the competition pressure it
faces from inactive positions. Suppose an extreme case where
only one position $j$ is active, so $A=\{j\}$. Then $Q_A=1$, and
\[
H_{\{j\}}^{(m)}
=
c_j^{(m)}
-
\beta\sum_{k\neq j}w_{jk}c_k^{(m)}.
\]
Thus the condition $\Xi\ge\tau_I$ requires
\[
c_j^{(m)}
-
\beta\sum_{k\neq j}w_{jk}c_k^{(m)}
\ge
\sigma\left[
\bar z(\omega^I)-\frac{1}{\bar z(\omega^I)}
\right]
\]
for every assignment $m$ and every position $j$. It means that the inactive positions
do not create enough benchmark through their effective advantage, for $j$ to exert a very high effort.

Particularly, the condition $\Xi\ge\tau_I$ is more likely to hold when $\sigma$
is very small. Since the relevant threshold in the condition is scaled by $\sigma$ through
$\sigma Q_A\tau_I$. Intuitively, a small $\sigma$ means that success is determined with
little noise, so effort translates more predictably into performance comparisons, making the planner not need to rely on very low evaluation
weights and corner equilibria to generate incentives.

The second sufficient condition is more conservative but uses a direct upper
bound on the value of any corner profile below $\omega^I$. Define
\[
\Lambda_A^{(m)}
:=
\mathbf 1_A^\top (I-\beta W_{AA})^{-1}\beta W_{A,-A}c_{-A}^{(m)}
-
\mathbf 1_A^\top c_A^{(m)}
=
-H_A^{(m)}.
\]
Let
\[
\Lambda_{\max}
:=
\max\left\{0,\max_{m\in\mathcal M}\max_{\varnothing\neq A\subsetneq J}
\Lambda_A^{(m)}\right\}.
\]
Also define
\[
V^I
:=
\max\left\{
\mathcal{W}(\omega_+^\ast,m^+),
\mathcal{W}(\omega_-^\ast,m^-)
\right\}.
\]

\begin{proposition}\label{prop:below-omegaI-Lambda}
Maintain Assumptions~\ref{ass:network-stability}, \ref{ass:alpha-bound}, and
\ref{ass:interior-candidate-ordering}. Suppose $\omega^I\le s$. If
\[
(1-s)B^+
+
s\left[
\frac{\sigma Q}{\underline\omega}\bar z(\omega^I)
+
\frac{1-\underline\omega}{\underline\omega}\Lambda_{\max}
\right]
\le
V^I,
\]
then allowing the planner to choose $\omega<\omega^I$ does not change the optimal
policy characterized in Proposition~\ref{prop:omega-cutoff}.
\end{proposition}

\begin{proof}
Again, $\omega\le\underline\omega$ is dominated by an active-contest policy by
Lemma~\ref{lem:collapse-threshold}. Consider
$\omega\in(\underline\omega,\omega^I)$ and any equilibrium with assignment $m$
and active set $A$. Its planner value is
\[
\mathcal W_A(\omega,m)
=
s\left[
\frac{\sigma Q_A}{\omega}\bar z(\omega)
+
\frac{1-\omega}{\omega}\Lambda_A^{(m)}
\right]
+
(1-s)b^\top a^{(m)}.
\]
Since $b^\top a^{(m)}\le B^+$, $Q_A\le Q$, $\bar z(\omega)\le \bar z(\omega^I)$,
$\omega>\underline\omega$, and
$\Lambda_A^{(m)}\le \Lambda_{\max}$, we have
\[
\mathcal W_A(\omega,m)
\le
(1-s)B^+
+
s\left[
\frac{\sigma Q}{\underline\omega}\bar z(\omega^I)
+
\frac{1-\underline\omega}{\underline\omega}\Lambda_{\max}
\right].
\]
By the maintained condition, the right-hand side is no larger than $V^I$.
Therefore no equilibrium with $\omega<\omega^I$ can give the planner a value
above the best value already attainable in $[\omega^I,1]$. Hence allowing
$\omega<\omega^I$ does not change the optimal policy characterized in
Proposition~\ref{prop:omega-cutoff}.
\end{proof}

\subsection{Proof of Corollary~\ref{cor:sstar-comparative}}\label{app:proof-corollary}

\begin{proof}
By Assumption~\ref{ass:interior-candidate-ordering},
\[
\omega^I
\le
\widehat\omega(B^-)
\le
\widehat\omega(B^+)
\le
1. \]
Moreover, since $\underline\omega e^{1/2}\le \omega^I$, we have
\[
\bar z(\omega^I)
=
\sqrt{2\log\left(\frac{\omega^I}{\underline\omega}\right)}
\ge 1.
\]
Together with $\omega^I\le \widehat\omega(B^-)$, this implies
\[
z(B^-)=\bar z(\widehat\omega(B^-))
\ge
\bar z(\omega^I)
\ge 1.
\]

By Proposition~\ref{prop:omega-cutoff}, under $B^+>B^-$,
\[
\widetilde T(B^-)-\widetilde T(B^+)>0
\]
and
\[
s^\star
=
\frac{B^+-B^-}{\widetilde T(B^-)-\widetilde T(B^+)},
\]
where, in the interior case,
\[
\widetilde T(B)
=
n\alpha+\frac{\sigma Q}{\widehat\omega(B)z(B)}.
\]

We first consider the effect of $Q$. By the envelope theorem, differentiating $\widetilde T(B)$ with
respect to $Q$ gives
\[
\frac{\partial \widetilde T(B)}{\partial Q}
=
\frac{\sigma z(B)}{\widehat\omega(B)}.
\]
Therefore
\[
\frac{\partial}{\partial Q}
\left[
\widetilde T(B^-)-\widetilde T(B^+)
\right]
=
\sigma
\left[
\frac{z(B^-)}{\widehat\omega(B^-)}
-
\frac{z(B^+)}{\widehat\omega(B^+)}
\right].
\]
Since $B^+> B^-$, we have $z(B^+)> z(B^-)$. Also,
\[
\widehat\omega(B)
=
\underline\omega\exp\left(\frac{z(B)^2}{2}\right),
\]
so
\[
\frac{z(B)}{\widehat\omega(B)}
=
\frac{1}{\underline\omega}
z(B)\exp\left(-\frac{z(B)^2}{2}\right).
\]
The function $z\exp(-z^2/2)$ is weakly decreasing for $z\ge1$. Since
$z(B^-)\ge1$, we obtain
\[
\frac{z(B^-)}{\widehat\omega(B^-)}
\ge
\frac{z(B^+)}{\widehat\omega(B^+)},
\]
and hence
\[
\frac{\partial}{\partial Q}
\left[
\widetilde T(B^-)-\widetilde T(B^+)
\right]
\ge0.
\]
Since $B^+-B^-$ does not depend on $Q$, it follows that $s^\star$ is weakly
decreasing in $Q$.

Next consider $B^+$, holding $B^-$ fixed. Since
\[
\frac{\partial \widetilde T(B)}{\partial B}
=
-\frac{1}{\widehat\omega(B)},
\]
we have
\[
\widetilde T(B^-)-\widetilde T(B^+)
=
\int_{B^-}^{B^+}\frac{1}{\widehat\omega(B)}\,dB.
\]
Because $\widehat\omega(B)$ is increasing in $B$, the integrand
$1/\widehat\omega(B)$ is decreasing in $B$. Differentiating $s^\star$ with
respect to $B^+$ gives
\[
\frac{\partial s^\star}{\partial B^+}
=
\frac{
\left[\widetilde T(B^-)-\widetilde T(B^+)\right]
-
\dfrac{B^+-B^-}{\widehat\omega(B^+)}
}{
\left[\widetilde T(B^-)-\widetilde T(B^+)\right]^2
}.
\]
Since
\[
\widetilde T(B^-)-\widetilde T(B^+)
=
\int_{B^-}^{B^+}\frac{1}{\widehat\omega(B)}\,dB
\ge
\frac{B^+-B^-}{\widehat\omega(B^+)},
\]
we obtain
\[
\frac{\partial s^\star}{\partial B^+}\ge0.
\]
This proves the result.
\end{proof}

\subsection{Proof of Proposition~\ref{prop:stability-diagonal}}\label{app:proof-stability-diagonal}
\begin{proof}
Given $\omega$, consider two agents $i\neq i'$ assigned to $j:=m(i)$ and
$k:=m(i')$. By \eqref{eq:Vi-closed-omega}, the common terms cancel in pairwise
comparisons. Therefore,
\[
\frac{1}{v}\Big(
V_i(\omega,m^{ii'})+V_{i'}(\omega,m^{ii'})
-
V_i(\omega,m)-V_{i'}(\omega,m)
\Big)
=
\frac{1-\omega}{\omega}
\Big(
\theta^{(m^{ii'})}_j+\theta^{(m^{ii'})}_k
-\theta^{(m)}_j-\theta^{(m)}_k
\Big).
\]
Using the notation in the main text,
\[
\theta^{(m)}_j = M_{jj}a_i + M_{jk}a_{i'} + R_j,
\qquad
\theta^{(m)}_k = M_{kj}a_i + M_{kk}a_{i'} + R_k,
\]
and
\[
\theta^{(m^{ii'})}_j = M_{jj}a_{i'} + M_{jk}a_i + R_j,
\qquad
\theta^{(m^{ii'})}_k = M_{kj}a_{i'} + M_{kk}a_i + R_k,
\]
where $R_j$ and $R_k$ are unchanged by the swap. Hence
\begin{align*}
&\theta^{(m^{ii'})}_j+\theta^{(m^{ii'})}_k
-\theta^{(m)}_j-\theta^{(m)}_k \\
&\qquad =
(M_{kk}-M_{jj})(a_i-a_{i'})
+
(M_{jk}-M_{kj})(a_i-a_{i'}).
\end{align*}
If $M$ is symmetric, then $M_{jk}=M_{kj}$, and therefore
\[
\frac{1}{v}\Big(
V_i(\omega,m^{ii'})+V_{i'}(\omega,m^{ii'})
-
V_i(\omega,m)-V_{i'}(\omega,m)
\Big)
=
\frac{1-\omega}{\omega}
(M_{kk}-M_{jj})(a_i-a_{i'}).
\]

If $\omega=1$, the expression is zero for every pair. Hence no pair has a
strictly profitable swap, and every assignment is pairwise stable.

Now suppose $\omega<1$. Since $\omega>0$, the factor $(1-\omega)/\omega$ is
strictly positive. Thus a swap is profitable if and only if
\[
(M_{kk}-M_{jj})(a_i-a_{i'})>0,
\]
or equivalently, if and only if
\[
(a_i-a_{i'})(M_{jj}-M_{kk})<0.
\]
Therefore, no profitable swap exists if and only if
\[
(a_i-a_{i'})(M_{jj}-M_{kk})\ge0
\]
for every pair $i\neq i'$. This proves \eqref{eq:stable-diagonal-ineq}.

The ordering statement follows immediately. After ordering abilities so that
$a_{i(1)}\le\cdots\le a_{i(n)}$, the pairwise inequalities require the assigned
diagonal scores to be weakly increasing in the same order, up to arbitrary
permutations within tied blocks.
\end{proof}
\subsection{Proof of Lemma~\ref{lem:stability-loss}}\label{app:proof-stability-loss}

\begin{proof}
We first record the branch value when the relevant stationary point is
feasible. For a branch with assignment value \(B\), define
\[
V_B(\omega;Q)
=
B+s\left[
n\alpha+
\frac{\sigma Q\bar z(\omega)-B-n\alpha}{\omega}
\right].
\]
If the branch selects its interior stationary point \(\widehat\omega(B)\), then
its optimized branch value is
\[
V_B^\ast(Q)
:=
V_B(\widehat\omega(B);Q)
=
B+s\widetilde T(B),
\]
where
\[
\widetilde T(B)
:=
n\alpha+
\frac{\sigma Q}{\widehat\omega(B)z(B)}.
\]

To see this, note that at the interior maximizer \(\widehat\omega(B)\), we have
\[
\bar z(\widehat\omega(B))=z(B),
\]
and the first-order condition implies
\[
z(B)-\frac{1}{z(B)}
=
\frac{B+n\alpha}{\sigma Q}.
\]
Hence
\[
B+n\alpha
=
\sigma Q\left(z(B)-\frac{1}{z(B)}\right),
\]
so
\[
\sigma Qz(B)-B-n\alpha
=
\frac{\sigma Q}{z(B)}.
\]
Substituting this into \(V_B(\widehat\omega(B);Q)\) gives
\[
V_B^\ast(Q)
=
B+s\left[
n\alpha+
\frac{\sigma Q}{\widehat\omega(B)z(B)}
\right]
=
B+s\widetilde T(B).
\]

The cutoff characterization guarantees that this is the relevant optimized
branch value in the cases used below. In particular,
\[
s<s^\star
\quad \Longrightarrow \quad
V^{FB}(Q)=V_{B^+}^\ast(Q)=B^+ +s\widetilde T(B^+),
\]
whereas
\[
s>s^\star
\quad \Longrightarrow \quad
V^{FB}(Q)=V_{B^-}^\ast(Q)=B^- +s\widetilde T(B^-).
\]
Analogously,
\[
s<s^{S,\star}
\quad \Longrightarrow \quad
V^{SB}(Q)=V_{B^{S,+}}^\ast(Q)=B^{S,+}+s\widetilde T(B^{S,+}),
\]
and
\[
s>s^{S,\star}
\quad \Longrightarrow \quad
V^{SB}(Q)=V_{B^{S,-}}^\ast(Q)=B^{S,-}+s\widetilde T(B^{S,-}).
\]
At equality, the corresponding two branches yield the same value.

Therefore, subtracting \(V^{SB}(Q)\) from \(V^{FB}(Q)\) gives four cases:
\[
\Delta(Q)
=
\begin{cases}
\big(B^+-B^{S,+}\big)
+s\big[\widetilde T(B^+)-\widetilde T(B^{S,+})\big],
& \text{if } s<s^\star \text{ and } s<s^{S,\star},\\[0.9em]
\big(B^+-B^{S,-}\big)
+s\big[\widetilde T(B^+)-\widetilde T(B^{S,-})\big],
& \text{if } s<s^\star \text{ and } s>s^{S,\star},\\[0.9em]
\big(B^--B^{S,+}\big)
+s\big[\widetilde T(B^-)-\widetilde T(B^{S,+})\big],
& \text{if } s>s^\star \text{ and } s<s^{S,\star},\\[0.9em]
\big(B^--B^{S,-}\big)
+s\big[\widetilde T(B^-)-\widetilde T(B^{S,-})\big],
& \text{if } s>s^\star \text{ and } s>s^{S,\star}.
\end{cases}
\]
This proves the lemma.
\end{proof}

\subsection{Proof of Lemma~\ref{lem:Q-cutoff}}\label{app:proof-lem-Q-cutoff}

\begin{proof}
Under the interior candidate ordering,
\[
\widetilde T(B;Q)
=
n\alpha+
\frac{\sigma Q}{\widehat\omega(B;Q)z(B;Q)}.
\]
From this, one obtains
\[
\frac{\partial}{\partial B}
\left[
\frac{\sigma Q}{\widehat\omega(B;Q)z(B;Q)}
\right]
=
-\frac{1}{\widehat\omega(B;Q)}.
\]
To see this, differentiate the identity
\[
z(B;Q)-\frac{1}{z(B;Q)}
=
\frac{B+n\alpha}{\sigma Q}
\]
with respect to \(B\), holding \(Q\) fixed. This gives
\[
\left(1+\frac{1}{z(B;Q)^2}\right)
\frac{\partial z(B;Q)}{\partial B}
=
\frac{1}{\sigma Q}.
\]
Hence
\[
\frac{\partial z(B;Q)}{\partial B}
=
\frac{z(B;Q)^2}{\sigma Q\left(1+z(B;Q)^2\right)}.
\]

Now recall that
\[
\widehat\omega(B;Q)
=
\underline\omega
\exp\left(\frac{z(B;Q)^2}{2}\right).
\]
Therefore
\[
\frac{\partial \widehat\omega(B;Q)}{\partial z}
=
\widehat\omega(B;Q)z(B;Q).
\]
It follows that
\[
\frac{\partial}{\partial z}
\left[
\frac{\sigma Q}{\widehat\omega(B;Q)z(B;Q)}
\right]
=
-\sigma Q
\frac{1+z(B;Q)^2}
{\widehat\omega(B;Q)z(B;Q)^2}.
\]
Using the chain rule,
\[
\frac{\partial}{\partial B}
\left[
\frac{\sigma Q}{\widehat\omega(B;Q)z(B;Q)}
\right]
=
-\sigma Q
\frac{1+z(B;Q)^2}
{\widehat\omega(B;Q)z(B;Q)^2}
\cdot
\frac{z(B;Q)^2}
{\sigma Q\left(1+z(B;Q)^2\right)}.
\]
After cancellation, this becomes
\[
\frac{\partial}{\partial B}
\left[
\frac{\sigma Q}{\widehat\omega(B;Q)z(B;Q)}
\right]
=
-\frac{1}{\widehat\omega(B;Q)}.
\]
Therefore
\[
\widetilde T(B^-;Q)-\widetilde T(B^+;Q)
=
\int_{B^-}^{B^+}
\frac{1}{\widehat\omega(B;Q)}\,dB.
\]
Hence
\[
s^\star(Q)
=
\frac{B^+-B^-}
{\int_{B^-}^{B^+}
\frac{1}{\widehat\omega(B;Q)}\,dB},
\]
so \(s=s^\star(Q)\) is equivalent to
\[
\Phi(Q)=\frac{1}{s}.
\]

It remains to show that this equation has a unique positive solution. For each
\(B\), \(K(B;Q)=(B+n\alpha)/(\sigma Q)\) is strictly decreasing in \(Q\). Since
\(z(B;Q)\) is strictly increasing in \(K(B;Q)\), \(z(B;Q)\) is strictly decreasing
in \(Q\). Hence
\[
\widehat\omega(B;Q)
=
\underline\omega\exp\left(\frac{z(B;Q)^2}{2}\right)
\]
is strictly decreasing in \(Q\), and therefore
\(1/\widehat\omega(B;Q)\) is strictly increasing in \(Q\). Thus \(\Phi(Q)\) is
strictly increasing in \(Q\).

Moreover,
\[
\lim_{Q\to0}\Phi(Q)=0,
\]
while
\[
\lim_{Q\to\infty}\Phi(Q)
=
\frac{1}{\underline\omega e^{1/2}}.
\]
Since \(\underline\omega e^{1/2}<s\), we have
\[
0<\frac{1}{s}<\frac{1}{\underline\omega e^{1/2}}.
\]
Therefore there exists a unique \(Q^\star>0\) satisfying
\(\Phi(Q^\star)=1/s\). Since \(\Phi(Q)\) is strictly increasing, the equivalence
between \(Q<Q^\star\) and \(s<s^\star(Q)\), and between \(Q>Q^\star\) and
\(s>s^\star(Q)\), follows immediately.
\end{proof}

\subsection{Proof of Proposition~\ref{prop:delta-monotone-Q}}\label{app:proof-delta-monotone-Q}

\begin{proof}
Since \(\underline\omega e^{1/2}<\omega^I\le s\), Lemma~\ref{lem:Q-cutoff}
applies. Moreover, the monotonicity statement below is understood on intervals
of \(Q\)-values on which Assumptions~\ref{ass:network-stability}--\ref{ass:interior-candidate-ordering}
continue to hold.

By Assumption~\ref{ass:interior-candidate-ordering},
\[
\omega^I\le \widehat\omega(B^-)\le \widehat\omega(B^+)\le 1.
\]
Since $\underline\omega e^{1/2}< \omega^I$, we have
\[
\bar z(\omega^I)
=
\sqrt{2\log\left(\frac{\omega^I}{\underline\omega}\right)}
\ge 1.
\]
Together with $\omega^I\le \widehat\omega(B^-)$, this implies
\[
z(B^-)\ge \bar z(\omega^I)\ge 1.
\]
Since $z(B)$ is increasing in $B$, we have $z(B)\ge1$ for every
$B\in[B^-,B^+]$.

For any branch value $B$, recall that
\[
\frac{\partial \widetilde T(B)}{\partial Q}
=
\frac{\sigma z(B)}{\widehat\omega(B)}.
\]
and
\[
\frac{z(B)}{\widehat\omega(B)}
=
\frac{1}{\underline\omega}
z(B)\exp\left(-\frac{z(B)^2}{2}\right).
\]
Since the function $z\exp(-z^2/2)$ is weakly decreasing for $z\ge1$,
\[
B_1\ge B_0
\quad\Longrightarrow\quad
\frac{\partial \widetilde T(B_1)}{\partial Q}
\le
\frac{\partial \widetilde T(B_0)}{\partial Q}.
\]
Therefore,
\[
B_1\ge B_0
\quad\Longrightarrow\quad
\frac{\partial}{\partial Q}
\left[
\widetilde T(B_1)-\widetilde T(B_0)
\right]
\le0.
\]

Now suppose \(Q<Q^\star\). By Lemma~\ref{lem:Q-cutoff}, this is equivalent to
\(s<s^\star(Q)\). Hence the unconstrained problem chooses the positive branch:
\[
V^{FB}(Q)=B^+ + s\widetilde T(B^+;Q).
\]
The stable constrained problem can choose either the stable positive branch
\(B^{S,+}\) or the stable negative branch \(B^{S,-}\). Therefore
\[
\Delta(Q)
=
\min_{B^S\in\{B^{S,+},B^{S,-}\}}
\left\{
(B^+-B^S)
+
s\left[\widetilde T(B^+;Q)-\widetilde T(B^S;Q)\right]
\right\}.
\]
For each \(B^S\in\{B^{S,+},B^{S,-}\}\), we have \(B^S\le B^+\). Hence
\[
\frac{\partial}{\partial Q}
\left[
\widetilde T(B^+;Q)-\widetilde T(B^S;Q)
\right]
\le 0.
\]
Thus each candidate loss is weakly decreasing in \(Q\). Since the minimum of
weakly decreasing functions is weakly decreasing, \(\Delta(Q)\) is weakly
decreasing in \(Q\).

Now suppose \(Q>Q^\star\). By Lemma~\ref{lem:Q-cutoff}, this is equivalent to
\(s>s^\star(Q)\). Hence the unconstrained problem chooses the negative branch:
\[
V^{FB}(Q)=B^- + s\widetilde T(B^-;Q).
\]
Therefore
\[
\Delta(Q)
=
\min_{B^S\in\{B^{S,+},B^{S,-}\}}
\left\{
(B^- -B^S)
+
s\left[\widetilde T(B^-;Q)-\widetilde T(B^S;Q)\right]
\right\}.
\]
For each \(B^S\in\{B^{S,+},B^{S,-}\}\), we have \(B^S\ge B^-\). Hence
\[
\frac{\partial}{\partial Q}
\left[
\widetilde T(B^-;Q)-\widetilde T(B^S;Q)
\right]
\ge 0.
\]
Thus each candidate loss is weakly increasing in \(Q\). Since the minimum of
weakly increasing functions is weakly increasing, \(\Delta(Q)\) is weakly
increasing in \(Q\).

This proves both claims.
\end{proof}

\end{document}